\documentclass[a4paper,UKenglish,cleveref, autoref, thm-restate]{lipics-v2021}
\usepackage{tikz-cd}
\usepackage{enumerate}
\usepackage{multicol}
\usepackage{hyperref}
\usepackage{adjustbox}
\usepackage{tikz}
\usepackage{float}
\usepackage{ifthen}
\usepackage{booktabs}
\usetikzlibrary{graphs.standard}
\usetikzlibrary{arrows.meta}
\usetikzlibrary{decorations.pathreplacing, calc, shapes.geometric}
\usepackage{graphicx}

\theoremstyle{plain}
\newtheorem{fact}[theorem]{Fact}
\crefname{fact}{Fact}{Facts}

\definecolor{pkc}{RGB}{31,110,190}   
\definecolor{ftc}{RGB}{31,140,90}    
\definecolor{brc}{RGB}{200,60,50}    
\definecolor{ntc}{RGB}{90,90,95}     

\newcommand{\del}{\mathrm{del}}
\newcommand{\w}{\mathbf{w}}
\providecommand{\VCdim}{\operatorname{VCdim}}

\newcommand{\ind}{\mathrm{ind}}
\newcommand{\forkindep}[1][]{%
  \mathrel{
    \mathop{
      \vcenter{
        \hbox{\oalign{\noalign{\kern-.3ex}\hfil$\vert$\hfil\cr
              \noalign{\kern-.7ex}
              $\smile$\cr\noalign{\kern-.3ex}}}
      }
    }\displaylimits_{#1}
  }
}

\newcommand{\flip}{\mathrm{flip}}
\newcommand{\rk}{\mathsf{rk}}

\newcommand{\td}{\mathrm{td}}

\newcommand{\dist}{\mathrm{dist}}

\newcommand{\C}{\mathcal{C}}

\newcommand{\Ind}[1]
{#1\setbox0=\hbox{$#1x$}\kern\wd0\hbox to 0pt{\hss$#1\mid$\hss} \lower.9\ht0\hbox to 0pt{\hss$#1\smile$\hss}\kern\wd0}

\newcommand{\notind}[1]
{#1\setbox0=\hbox{$#1x$}\kern\wd0
\hbox to 0pt{\mathchardef\nn=12854\hss$#1\nn$\kern1.4\wd0\hss}
\hbox to 0pt{\hss$#1\mid$\hss}\lower.9\ht0 \hbox to 0pt{\hss$#1\smile$\hss}\kern\wd0}

\newcommand{\N}{\mathbb{N}}

\newcommand{\R}{\mathbb{R}}
\newcommand{\Gaif}{\mathsf{Gaif}}

\newcommand{\Bcal}{\ensuremath{\mathcal{B}}}

\newcommand{\Fcal}{\ensuremath{\mathcal{F}}}

\newcommand{\Ocal}{\ensuremath{\mathcal{O}}}
\newcommand{\Pcal}{\ensuremath{\mathcal{P}}}

\newcommand{\Rcal}{\ensuremath{\mathcal{R}}}

\newcommand{\Xcal}{\ensuremath{\mathcal{X}}}
\newcommand{\Ycal}{\ensuremath{\mathcal{Y}}}

\providecommand{\C}{\mathcal{C}}
\providecommand{\N}{\mathbb{N}}
\providecommand{\R}{\mathbb{R}}
\providecommand{\Fcal}{\mathcal{F}}
\providecommand{\Bcal}{\mathcal{B}}
\providecommand{\Pcal}{\mathcal{P}}
\providecommand{\Ocal}{\mathcal{O}}
\providecommand{\Rcal}{\mathcal{R}}
\providecommand{\ind}{\operatorname{ind}}
\providecommand{\flip}{\operatorname{flip}}
\providecommand{\dist}{\operatorname{dist}}
\providecommand{\rk}{\operatorname{rk}}
\providecommand{\td}{\mathrm{td}}

\providecommand{\Gaif}{\operatorname{Gaif}}
\providecommand{\forkindep}{\mathrel{\raisebox{0.15ex}{$\smile$}\!\!\!\!\raisebox{-0.9ex}{$|$}}}
\declaretheorem[name=Question,sibling=theorem]{question}

\title{Flip-packability: uniform characterisations of tame graph classes}
\titlerunning{Flip-packability: uniform characterisations of tame graph classes}

\author{Ioannis Eleftheriadis}{Department of Computer Science, University of Oxford \and \url{https://www.cs.ox.ac.uk/people/ioannis.eleftheriadis/} }{ioannis.eleftheriadis@cs.ox.ac.uk}{https://orcid.org/0000-0003-4764-8894}{}
\authorrunning{Ioannis Eleftheriadis}
\Copyright{Ioannis Eleftheriadis}



\funding{The work was supported by UKRI EP/X024431/1.}

\EventEditors{}
\EventNoEds{0}
\EventLongTitle{}
\EventShortTitle{}
\EventAcronym{}
\EventYear{}
\EventDate{}
\EventLocation{}
\EventLogo{}
\SeriesVolume{}
\ArticleNo{}

\begin{document}
\hideLIPIcs
\nolinenumbers
\maketitle

\begin{abstract}
A class of graphs is \emph{monadically dependent} if one cannot encode all graphs in coloured graphs from the class using a fixed first-order formula, and \emph{monadically stable} if one cannot even encode arbitrarily long linear orders. Bonnet et al.\ (ICALP 2025) characterised monadic dependence by \emph{flip-separability}: for every vertex weighting, boundedly many flips --- complementations of the adjacency relation within a vertex subset --- make every ball of radius $r$ carry at most an $\varepsilon$-fraction of the weight, so that every set carrying an $\varepsilon$-fraction has two elements pulled apart.

We introduce \emph{flip-packability}: boundedly many graphs, each obtained from the input by boundedly many flips and all determined by the weighting before any set is presented, such that every set carrying an $\varepsilon$-fraction of the weight has $m$ elements pairwise far apart in one of them. The number of flips producing each graph depends on the radius alone; only the number of graphs depends on $\varepsilon$ and $m$. We prove that a class of graphs is flip-packable if and only if it is monadically stable, and $2$-flip-packable, that is, flip-packable with $m=2$, if and only if it is monadically dependent. The passage from two scattered elements to $m$ is thus exactly what separates the two notions. For monadically stable classes we show that the flipped graphs can be computed from the weighting
in cubic time.

Varying the three parameters of the definition --- the sparsifying operation, the radius, and the number $m$ of elements scattered --- produces eight known characterisations of sparse and dense graph classes from the same template. In each case $m$ separates a depth-like notion from its width-like relaxation: treedepth from treewidth, shrubdepth from cliquewidth, and monadic stability from monadic dependence.
\end{abstract}

\section{Introduction}

Deciding whether a first-order sentence holds in a given finite graph is in general hard: the
problem is complete for the parameterised complexity class $\mathrm{AW}[*]$. It becomes tractable
on restricted classes of inputs, and the search for the exact boundary of tractability has led to a
substantial body of work over the last fifteen years, combining structural graph theory with model
theory and parameterised complexity; see \cite{pilipczuk-lens,siebertz-vigny} for surveys. For
classes of sparse graphs the answer is known and is remarkably clean: if $\C$ is closed under
subgraphs, then first-order model checking is fixed-parameter tractable on $\C$ if and only if $\C$
is \emph{nowhere dense}~\cite{deciding}. This already covers a great deal, including planar graphs,
graphs of bounded treewidth or bounded degree, and every class excluding a fixed minor. Beyond
sparsity, the picture has come into focus only recently. The conjectured frontier is \emph{monadic
dependence}: for classes closed under \textit{induced} subgraphs it is necessary for
tractability~\cite{flipbreak}, and it is known to be sufficient in the more restrictive setting of
\emph{monadic stability}~\cite{dreier2023firstorder}.

Both notions originate in model theory \cite{baldwin1985second, shelah-hanf} and are described by what they
forbid. A class is monadically dependent if one cannot produce every finite graph from members of
$\C$ by colouring the vertices with a bounded number of colours and then applying a fixed
first-order formula; it is monadically stable if one cannot even produce every finite linear order
this way. Colours are precisely what makes the notions robust: both are preserved under exactly the
transformations to which first-order model checking is insensitive.

Algorithms do not manipulate colourings and formulas directly, however. Instead, they rely on
combinatorial descriptions of these classes, which supply decompositions to recurse on, and much
recent effort has gone into finding such descriptions
\cite{indiscernibles,flippergame,flipbreak,separability}. Interestingly, these are often modelled
on properties of sparse graphs. For example, an archetypal sparse statement is that every
$n$-vertex tree has a \emph{centroid}, a vertex whose removal leaves components of at most $n/2$
vertices, and so, by the same argument on a tree decomposition, graphs of treewidth $k$ have
balanced separators of $k+1$ vertices in the same sense. Iterating, for every $\varepsilon>0$
boundedly many deletions leave no component with more than an $\varepsilon$-fraction of the
vertices. Beyond bounded treewidth such global separators need not exist, but a \emph{local} form
survives on every nowhere dense class: for every radius $r$ and $\varepsilon>0$, boundedly many
deletions leave no ball of radius $r$ with more than an $\varepsilon$-fraction of the vertices
\cite{nevsetvril2016structural}.

Deleting vertices is useless on dense graphs --- a clique stays a clique --- and the operation
that replaces it is the \emph{flip}: complementing the adjacency relation within a subset of the
vertices. A single flip turns a clique into an independent set, and flips in general sparsify dense
graphs in a controlled and reversible way. Bonnet et al.\ \cite{separability} recently established
the dense form of the local statement above, and showed that it characterises monadic dependence:
a class is \emph{flip-separable} if for every $r$ and $\varepsilon$, for every graph in the class
and every weighting of its vertices, boundedly many flips make every ball of radius $r$ carry at
most an $\varepsilon$-fraction of the weight, the balls around the few vertices that alone carry
more than that excepted. Equivalently, every set carrying an $\varepsilon$-fraction of the weight
and consisting of lighter vertices has two elements pulled apart at distance greater than $r$, since such a set cannot lie
inside a ball. Flip-separability is one of a family of characterisations of the two monadic
notions by flips, mirroring those of nowhere density by deletions, which we recall in
\Cref{sec:prelims}; it is the one this paper builds on.

\subparagraph*{Our contribution.}
Our starting point is a feature of flip-separability that is easy to pass over: the number of flips
it needs depends on $\varepsilon$ as well as on $r$, and grows as $\varepsilon$ shrinks. 
While this is unavoidable even in trees, a careful analysis of separators reveals a finer relationship between these quantities, allowing to move the dependence on
$\varepsilon$ somewhere less costly.

To illustrate, consider how a separator in a tree is found.
Given a set $X$ of at least $\varepsilon n$ vertices,
delete a centroid, see which component $X$ went into, delete that component's centroid, and so on;
the component holding $X$ halves each time, so $X$ is split within $\log(1/\varepsilon)$ steps. This
separator was built \emph{following} $X$. But the vertices it could ever use are the centroids of
the components that are heavy at some level of the recursion, and there are boundedly many of
those, so they can be listed before $X$ is seen. Deleting the whole list at once is the balanced
separator of the sparse statement above. The list says more than its union, though: presented with
any heavy $X$, a \emph{single} entry splits it, namely the centroid of the component in which the
procedure following $X$ would have stopped. Which entry depends on $X$; the list does not; and each
entry is one vertex, however small $\varepsilon$ is. The cost of $\varepsilon$ has moved from the
size of the separator to the length of the list.

This gives two elements of $X$ pulled apart, which is all that separability asks. Pulling $m$
elements apart is a different matter. The centroid step made progress by shrinking the region that
holds $X$, and it produced a split because $X$ cannot fit in a region of fewer than $\varepsilon n$
vertices. For $m$ elements in pairwise distinct components, $X$ must meet $m$ of the components
left by a single deleted vertex, and no choice of vertex can force this: on a path no vertex has
more than two branches, and $k$ deletions leave at most $k+1$ components. What makes it possible is
height. Root a tree of height $h$ and follow $X$ downwards: at each node, either $X$ meets $m$ of
the subtrees below, and deleting that node alone separates them, or $X$ is confined to fewer than
$m$ of them, one of which holds a $\frac1m$-fraction of $X$, and we descend there. After $i$ steps
the current subtree holds at least $\varepsilon n/m^i$ vertices of $X$; so if $n>m^h/\varepsilon$
the descent cannot reach a leaf and must therefore stop within $h$ steps at a node
below which $X$ meets $m$ subtrees. The candidate nodes are again boundedly many and listable in
advance. 
\Cref{fig:trees} shows both constructions. What they
share is the shape of the answer: a bounded list of small parameter sets, fixed before $X$, one of
which serves whatever $X$ arrives; the size of each set is fixed independently of $\varepsilon$
and $m$, and only the length of the list depends on them.

\begin{figure}[h]
\centering
\includegraphics[width=0.75\textwidth]{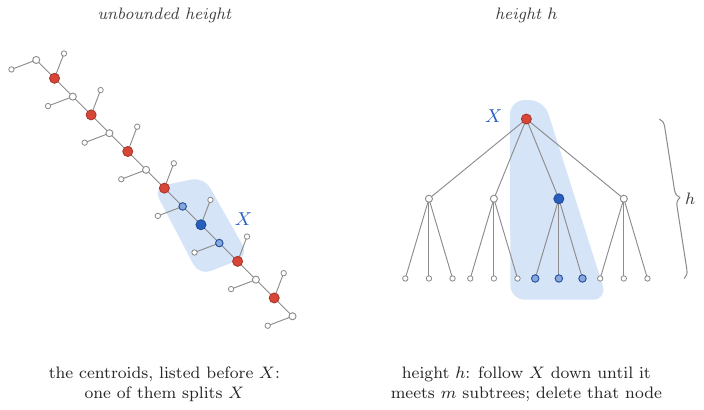}
\caption{\small Serving a heavy set $X$, shaded, in a tree. Red vertices are the candidates listed
before $X$ is seen; blue is the one that serves it. Left: the centroids of all heavy components,
one of which splits $X$. Right: in a tree of height $h$, $X$ followed downwards cannot end on a
leaf, so at some node it meets $m$ subtrees, and deleting that node scatters it.}
\label{fig:trees}
\end{figure}

\textit{Flip-packability}, the property this paper introduces, is the bounded-height construction
for dense graphs, with flips for deletions and balls of radius $r$ for components. Trees of unbounded height correspond to monadically dependent classes and trees of bounded
height to monadically stable ones; we return to this analogy below. Our main result is that, as
bounded height did, monadic stability buys the passage from two elements to $m$. Here a
\emph{weighting} assigns a non-negative weight to every vertex.

\begin{theorem}\label{thm:main-intro}
A class $\C$ of graphs is monadically stable if and only if it is \emph{flip-packable}: for every
radius $r$ there is $k$ such that for every $\varepsilon>0$ and every $m$ there are $\delta>0$ and
$t$ with the following property. For every $G\in\C$ and every weighting of $G$, there are $t$
graphs, each obtained from $G$ by at most $k$ flips, such that every set carrying an
$\varepsilon$-fraction of the weight, none of whose vertices carries a $\delta$-fraction, contains
$m$ vertices pairwise at distance more than $r$ in one of them.
\end{theorem}

\begin{figure}[ht]
\centering
\includegraphics[width=0.92\textwidth]{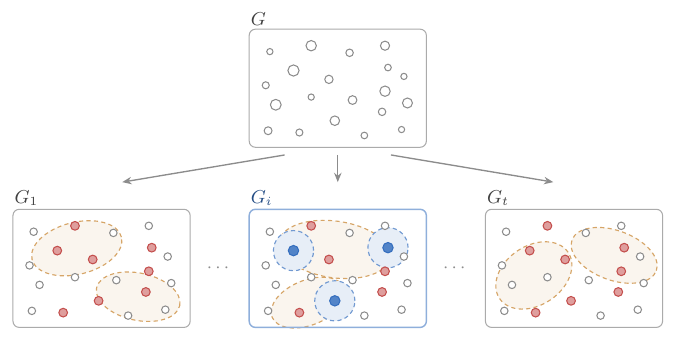}
\caption{\small Flip-packability. From $G$ and a weighting, drawn as the sizes of the vertices,
one computes $t$ graphs $G_1,\dots,G_t$, each obtained from $G$ by at most $k$ flips; the orange shaded
regions mark the flipped parts. Only afterwards is an $\varepsilon$-heavy set $X$ of
$\delta$-light vertices presented, drawn in red. One of the $t$ graphs then contains $m$ elements
of $X$ pairwise at distance more than $r$, drawn in blue. The
number $k$ of flips depends on $r$ alone; only $t$ and $\delta$ depend on $\varepsilon$ and $m$.}
\label{fig:packing-op}
\end{figure}

See \Cref{fig:packing-op} for an illustration of flip-packability. Notice that when $\varepsilon=1$ only the support of the weighting needs serving, and the
property essentially becomes \emph{flip-flatness} --- after boundedly many flips, every large
vertex set has a large $r$-scattered subset, the dense analogue of the \emph{uniform
quasi-wideness} of nowhere dense classes \cite{NOdM11} --- which characterises monadic stability
\cite{indiscernibles}. Put differently: flip-flatness scatters many elements of a single large set,
flip-separability scatters two elements of every dense subset, and flip-packability scatters many
elements of every dense subset, all from one family fixed in advance. So \Cref{thm:main-intro} is
not a new necessary condition but the statement that monadic stability implies a considerably more
uniform one than either property it descends from.

Moreover, using a
technique of Bonnet et al.\ \cite{separability}, the $t$ flipped graphs in the definition of flip-packability can essentially be merged
into a single one in which every heavy set of light vertices still has $m$ scattered elements. But
the merging is precisely what loses the information that the number of flips producing each graph
depends on the radius alone: the merged graph is obtained by a number of flips that grows with
$\varepsilon$ and $m$, and once that dependence is allowed the statement no longer separates
stability from dependence, since flip-separability with tolerance $\varepsilon/m$ already provides
it. \Cref{thm:main-intro} therefore rests on two features of the family at once, its two-tier shape
and the $m$ elements it scatters. It is natural to ask how much each contributes, and our second
result answers this by locating the case $m=2$ exactly.

\begin{theorem}\label{thm:2packing-intro}
A class $\C$ of graphs is monadically dependent if and only if it is \emph{$2$-flip-packable},
that is, satisfies the case $m=2$ of flip-packability.
\end{theorem}

So the two-tier form is already present in monadically dependent classes --- just as the list of
single centroids exists in every tree --- and what monadic stability adds is exactly the passage
from two scattered elements to $m$, with the number of flips per graph still independent of $m$.
The proof sandwiches $2$-flip-packability between monadic dependence and flip-separability. The
family for dependent classes is read off the proof of \cite{separability}, which compresses a
large family into a bounded one of exactly this shape and merges it into a single flip only at the
very end; conversely, the merging just described turns any such family into the flip of
flip-separability. So $2$-flip-packability is the two-tier form of flip-separability, and the two
are equivalent. 

The proof of \Cref{thm:main-intro} on the other hand is an induction on the length of the \emph{Flipper game} of Gajarsk\'y et al.\ \cite{flippergame}, a two-player localisation game in which one player repeatedly flips the graph and the other zooms into a ball, and whose bounded length characterises monadic stability; in the analogy above it plays the role of the height of the tree. At each round the arena is packed with disjoint balls of large weight, of which only boundedly many fit, and a heavy set either is already pulled apart by the flips played so far or concentrates inside one of the packed balls, where the game is a round closer to its end. 
Read
algorithmically, the result is a preprocessing statement. The family of flipped graphs plays the
role of an $\varepsilon$-net \cite{HW87}: it is computed from the weighting, before any query,
and then serves every sufficiently heavy and balanced query set. For monadically stable classes
this is effective, by driving the induction with the algorithmic Flipper strategy of
\cite{flippergame-full}.

\begin{theorem}\label{thm:algo-intro}
Let $\C$ be a monadically stable class of graphs, let $r$, $\varepsilon$ and $m$ be given, and let
$k$, $t$ and $\delta$ be the constants of \Cref{thm:main-intro} for them. There is an algorithm
that, given an $n$-vertex graph $G\in\C$ and a weighting of $G$ with rational weights, computes
in time $\Ocal_{\C,r,\varepsilon,m}(n^3)$ a family of $t$ graphs, each obtained from $G$ by at
most $k$ flips, such that every set carrying an $\varepsilon$-fraction of the weight, none of
whose vertices carries a $\delta$-fraction, contains $m$ vertices pairwise at distance more than
$r$ in one of them. Arithmetic on the weights is counted at unit cost.
\end{theorem}

In fact, $1/\delta$ and $t$ may be taken polynomial in $m$ and $1/\varepsilon$, with exponents
depending on $r$ alone (\Cref{cor:polynomial}); at $\varepsilon=1$ this recovers the polynomial
and algorithmic flip-flatness of \cite[Theorem~1.4]{indiscernibles}, and for general $\varepsilon$
it says that serving every dense subset, rather than the single large set flatness is handed,
costs only polynomially more. 
Moreover, the same algorithm computes in cubic time the single flip of flip-separability on every monadically stable class. By making effective the compression argument of \cite{separability}, we show that the same holds, with a larger exponent, on all monadically dependent classes (\Cref{thm:algo-dependent,cor:algo-sep}).

\subparagraph*{Everything is packing.}
\Cref{thm:main-intro,thm:2packing-intro} read $m$ as a dial: unbounded, it gives monadic stability;
fixed to $2$, monadic dependence. This places flip-packability in a wider picture. Two further
features of the definition can be varied independently of $m$: the sparsifying operation, which may
be flips or vertex deletions, and the radius, which may be finite or infinite, in which case balls
are connected components. Each of the eight resulting statements characterises a known notion.

\begin{theorem}\label{thm:everything}
For a class $\C$ of graphs, with no closure assumptions, the eight properties obtained from
flip-packability by fixing the operation, the radius, and whether $m$ is unbounded or fixed to $2$
characterise the notions in \Cref{tab:variants-intro}.
\end{theorem}

\begin{table}[htbp]
  \caption{The variants of flip-packability, obtained by varying the sparsifying operation, the
  radius, and the number of elements scattered.  
  Precise statements are in \Cref{sec:variants} and proofs in \Cref{app:variants} of the appendix.}
  \label{tab:variants-intro}
  \centering
  \begin{tabular}{@{}llll@{}}
    \toprule
    Operation & Radius & All $m$ & $m=2$ \\
    \midrule
    deletions & $r\in\N$  & nowhere dense & nowhere dense \\
    flips     & $r\in\N$  & monadically stable  & monadically dependent \\
    deletions & $\infty$  & bounded treedepth & bounded treewidth \\
    flips     & $\infty$  & bounded shrubdepth & bounded cliquewidth \\
    \bottomrule
  \end{tabular}
\end{table}

The tree example is the third row read across both columns: with deletions and components, the
list of centroids is the $m=2$ entry and the bounded-height construction the all-$m$ entry. For a
tree, bounded treewidth is automatic and bounded treedepth is bounded height, so the row divides
exactly where the analogy did. The scheme itself is not new. Dreier, M\"ahlmann and Toru\'nczyk
classify two properties along the same two axes --- flip-flatness, and \emph{flip-breakability},
under which after boundedly many flips every large vertex set splits into two large parts far from
each other, which characterises monadic dependence \cite{flipbreak} --- and prove the eight
characterisations that fill the table \cite[Sections~1 and 18]{flipbreak-full}. What is new is that
all eight are instances of a single definition. Flatness and breakability are essentially the two
values that $m$ can take, unbounded and bounded, so what were two two-dimensional families become
one three-dimensional family, and the relation between them becomes a parameter of the definition. While the two could be unified more cheaply, by asking for $m$ large parts pairwise far apart ---
flatness being the case of singleton parts and breakability that of two --- such a property would
still be handed the set it must split, whereas flip-packability is uniform in that it scatters every dense subset from
one family fixed in advance.

The eight statements are proved along the same two arguments, once the operation and the radius
are changed. In each case the packing property is first weakened to the corresponding variant of
flatness, by taking $\varepsilon=1$, or of breakability, by taking $m=2$ and splitting a large set
along the balls produced; these are the properties \cite[Section~18]{flipbreak-full} shows to
characterise the eight entries, so only the converses remain. For the left-hand column the
converse is essentially the induction used for \Cref{thm:main-intro}, ran on a suitable game rank. 
For the right-hand column there is no such rank, and the family is obtained instead by listing a balanced separator for every region heavy enough to hold $X$, at
every level of a recursion of depth $\Ocal(\log\frac1\varepsilon)$, much like the list of centroids of the tree analogy. 

\subparagraph*{Relational structures.}
Flip-packability makes sense for classes of structures in a finite relational signature, with the
flips of Przybyszewski and Toru\'nczyk \cite{flipfork}: a flip of a structure over a set of
parameters is a structure on the same domain, in a possibly different signature, that is
quantifier-free interdefinable with the original using those parameters. In \Cref{app:relational}
of the appendix we show that the characterisation of \Cref{thm:main-intro} extends to the
relational setting. The proof is the same, with the relational Flipper theorem of \cite{flipfork}
in place of \cite{flippergame}, modulo some additional properties of relational flips. 

\subparagraph*{Where the property comes from.}
Flip-packability was found by passing to the limit. In a companion paper \cite{companion} we show
that the infinitary forms of the flip characterisations --- ``large'' replaced by ``infinite'',
``boundedly many flips'' by ``finitely many'', and asked of every model of the theory of $\C$ ---
all characterise monadic stability, including those that characterise monadic dependence in the
finite. The limit form of flip-separability is particularly simple: for every non-principal
ultrafilter, some finite flip makes every $r$-ball null. Three parameters of the finitary property
are vacuous there --- the tolerance, the number of scattered elements, and the family of flips ---
and restoring them by hand is exactly flip-packability. The present paper is self-contained and
finitary.

\subparagraph*{Organisation.}
\Cref{sec:prelims} fixes notation and recalls the flip characterisations of the two monadic
notions. \Cref{sec:packing} introduces flip-packability and proves \Cref{thm:main-intro}.
\Cref{sec:backtosep} treats the case $m=2$ and proves \Cref{thm:2packing-intro}.
\Cref{sec:variants} surveys the variants of \Cref{thm:everything}, with proofs in
\Cref{app:variants}, and \Cref{sec:algorithmic} states the algorithmic results, with proofs in
\Cref{app:algorithmic}. \Cref{app:relational}
explains how the results extend to structures in a finite relational signature.

\section{Preliminaries}\label{sec:prelims}

Graphs are finite, simple and undirected. For a graph $G$ we write $N(v)$ for the neighbourhood of
$v$, $\dist_G(u,v)$ for the distance, $B^r_G(v):=\{u:\dist_G(u,v)\le r\}$ for the ball of radius
$r$, and $G-S$ for the graph obtained by deleting the vertices of $S$. A set $A$ is
\emph{$r$-independent} in $G$ if its elements are pairwise at distance greater than $r$, and
$\ind^r_G(X)$ denotes the largest size of an $r$-independent subset of $X$. For a set $V$ we write
$[V]^{\le k}$ for its subsets of size at most $k$. Distances are measured in $\N\cup\{\infty\}$,
and at radius $\infty$ we read ``distance greater than $\infty$'' as ``in distinct connected
components''. A class of graphs is any set of graphs; no closure assumptions are made.

A \emph{weighting} of $G$ is a function $\w:V(G)\to\R_{\ge0}$ that is not identically zero,
extended to sets by $\w(X):=\sum_{v\in X}\w(v)$. A set $X\subseteq V(G)$ is
\emph{$\varepsilon$-heavy} for $\w$ if $\w(X)\ge\varepsilon\cdot\w(V(G))$, and $\w$ is
\emph{$\delta$-balanced} if no vertex is $\delta$-heavy, that is, $\w(v)<\delta\cdot\w(V(G))$ for
all $v$. The \emph{indicator weighting} of a non-empty $A\subseteq V(G)$ is $\w_A:=\mathbf 1_A$; it
is $\delta$-balanced as soon as $|A|>1/\delta$, and a set is $\varepsilon$-heavy for it exactly
when it contains at least $\varepsilon|A|$ elements of $A$. Every notion below is invariant under
scaling $\w$, so in proofs we assume $\w(V(G))=1$, where $\varepsilon$-heavy reads
$\w(X)\ge\varepsilon$ and $\delta$-balanced reads $\w(v)<\delta$.

\subparagraph*{VC-dimension, monadic stability and dependence.}
We use the standard terminology of VC-theory for graphs. A graph $G$ \emph{shatters} a set
$A\subseteq V(G)$ if every $B\subseteq A$ is of the form $N(u)\cap A$ for some vertex $u$. The
\emph{VC-dimension} $\VCdim(G)$ of $G$ is the largest size of a set it shatters, and a class $\C$
has \emph{bounded VC-dimension} if $\VCdim(\C):=\sup\{\VCdim(G):G\in\C\}$ is finite.

By a \textit{vertex-coloured graph} we mean a relational structure on the vertex set
with a binary relation $E$ for adjacency and a unary relation for each colour. A transduction
$\mathsf T$ consists of a finite set of colours $\Sigma$ and a first-order formula $\varphi(x,y)$
in the language of $\Sigma$-coloured graphs. Given a graph $G$, the family $\mathsf T(G)$ consists
of all induced subgraphs of the graphs $H$ with $V(H)=V(G)$ for which some colouring of $G$, that
is, some choice of sets $U_C\subseteq V(G)$ for $C\in\Sigma$, makes $uv\in E(H)$ equivalent to
$\varphi(u,v)$ holding in the coloured graph, for all distinct $u,v\in V(H)$. For a class $\C$ we
put $\mathsf T(\C):=\bigcup_{G\in\C}\mathsf T(G)$, and $\C$ \emph{transduces} a class $\mathcal D$
if $\mathcal D\subseteq\mathsf T(\C)$ for some transduction $\mathsf T$.

A graph class is \emph{monadically dependent} if it does not transduce the class of all graphs,
and \emph{monadically stable} if it does not transduce the class of all \emph{half-graphs}, that
is, the graph-theoretic analogues of linear orders: the bipartite graphs with sides 
$a_1,\dots,a_n$ and $b_1,\dots,b_n$ and edges $a_ib_j$ for $i\le j$. Monadically stable classes are
monadically dependent, and monadically dependent classes have bounded VC-dimension, as every class of
unbounded VC-dimension transduces all graphs.

\subparagraph*{Flips.}
A \emph{flip} of a graph $G$ is the graph obtained by complementing the adjacency relation within
a set $A\subseteq V(G)$, or more generally between two sets $A,B\subseteq V(G)$; a graph is
\emph{obtained from $G$ by $k$ flips} if it results from at most $k$ such operations applied in
sequence. This is the notion of the introduction. As in \cite{separability}, it is more convenient
to view a sequence of flips as a single operation on a partition. For a partition $\Pcal$ of
$V(G)$, a \emph{$\Pcal$-flip} of $G$ is a graph obtained by choosing a set of unordered pairs of
parts, the two members of a pair possibly equal, and complementing the adjacency between the two
parts of every chosen pair. A graph obtained by $k$ flips is a $\Pcal$-flip for a partition into at
most $4^{k}$ parts, and a $\Pcal$-flip with $|\Pcal|\le p$ is obtained by at most $p^2$ flips, so
the two views are interchangeable up to a change of constants.

Following \cite{separability} and the works it builds on \cite{BonnetDGKMST22,flippergame}, we
work with partitions defined from a bounded set of parameter vertices. Fix $S\subseteq V(G)$. The
\emph{$S$-classes} of $G$ are the singletons $\{s\}$ for $s\in S$ and the sets
$\{v\in V(G)\setminus S : N(v)\cap S=A\}$ for $A\subseteq S$, so there are at most $2^{|S|}+|S|$ of
them; we write $\Pcal_S$ for the partition into $S$-classes and $[v]_S$ for the class of $v$. An
\emph{$S$-flip} of $G$ is a $\Pcal_S$-flip, that is, a graph $G'$ on $V(G)$ determined by a set
$\Pcal$ of unordered pairs of $S$-classes by
\[ uv\in E(G')\iff uv\in E(G)\ \oplus\ \{[u]_S,[v]_S\}\in\Pcal. \]
This is the \emph{$S$-definable flip} of \cite{BonnetDGKMST22}, and we write $\flip(G/S)$ for the
set of $S$-flips of $G$. Since there are at most $(2^{|S|}+|S|)^2$ pairs of classes,
\[ |\flip(G/S)|\;\le\;\Xi(|S|):=2^{(2^{|S|}+|S|)^2}, \]
a bound depending on $|S|$ alone, and $\Xi$ is monotone. Three properties are used throughout.
If $F\subseteq S$ then the $S$-classes refine the $F$-classes, so $\flip(G/F)\subseteq\flip(G/S)$.
Every $s\in S$ is isolated in some $S$-flip: $N(s)$ is a union of $S$-classes, and complementing
the pairs $\{\{s\},C\}$ for these classes $C$ removes every edge at $s$. Doing so for all $s\in S$
at once gives an $S$-flip in which $S$ is isolated and which agrees with $G-S$ elsewhere, so flips
subsume vertex deletions. 

Since it is usually irrelevant \emph{which} flip scatters a given set, we quantify over all of them
at once: the \emph{flip-metric over $S$}, introduced by Gajarsk\'y et al.\ \cite{flippergame}, is
given by
\[ \dist_S(u,v):=\max_{G'\in\flip(G/S)}\dist_{G'}(u,v). \]
As a maximum of finitely many metrics with values in $\N\cup\{\infty\}$, this is a valid metric; in
particular it satisfies the triangle inequality, which we use freely. We write $B^r_S(u)$ for its
balls, so that $B^r_S(u)=\bigcap_{G'\in\flip(G/S)}B^r_{G'}(u)$, and note that $B^r_S(s)=\{s\}$ for
$s\in S$ by the isolating flip, and that $\dist_F\le\dist_S$ for $F\subseteq S$. We write
$\ind^r_S(X)$ for the largest size of a subset of $X$ that is $r$-independent in the flip-metric
over $S$, as opposed to $\ind^r_{G'}(X)$ in a single flip $G'$. The two differ, but the flip-metric
is realised by at most $\Xi(|S|)$ flips, which lets us pass from one to the other by a Ramsey
argument.

\subparagraph*{Metric conversion.}
Definable flips are less powerful than partition flips, and the flip-metric is not the metric of
any one graph; still, both gaps close up to a constant factor in the radius. 

\begin{fact}[metric conversion]\label{fact:conversion}
Let $G$ be a graph.
\begin{enumerate}
\item \cite[Lemma~4]{separability} If $G$ has VC-dimension at most $d$, then for every partition
$\Pcal$ of $V(G)$ and every $\Pcal$-flip $H$ of $G$ there are $S\subseteq V(G)$ with
$|S|\le\Ocal(d|\Pcal|^2)$ and $H'\in\flip(G/S)$ with $B^r_{H'}(v)\subseteq B^{5r}_H(v)$ for all
$v\in V(G)$ and $r\in\N$; in particular every component of $H'$ is contained in one of $H$.
\item \cite[Lemma~8]{separability} For every $S\subseteq V(G)$ with $|S|\le k$ there are
$T\supseteq S$ with $|T|\le k+ (2^k+k)^2$ and $G'\in\flip(G/T)$ with $B^r_{G'}(v)\subseteq B^{6r}_S(v)$
for all $v\in V(G)$ and $r\in\N$, computable from $G$ and $S$ in time $\Ocal(4^kn^2)$.
\end{enumerate}
\end{fact}

Item (2) is \cite[Lemma~8]{separability} applied to the partition $\Pcal_S$, whose metric is
$\dist_S$: it yields a refinement $\Pcal'$ of $\Pcal_S$ and a $\Pcal'$-flip $G'$ with
$B^r_{G'}(v)\subseteq B^{6r}_S(v)$. Inspecting its proof, $\Pcal'$ refines $\Pcal_S$ by
neighbourhood types over a set $A$ of at most $|\Pcal_S|^2$ vertices, one for each ordered pair of
$S$-classes, so $\Pcal_{S\cup A}$ refines $\Pcal'$ and $G'$ is an $(S\cup A)$-flip. Bounded
VC-dimension only improves the bound on $|T|$, to $k^{\Ocal(d^2)}$ \cite[Lemma~5]{separability}.
Thus the flip-metric over a bounded set is always dominated by the metric of a single definable
flip over a bounded set, and on classes of bounded VC-dimension, which include all monadically
dependent classes, partition flips are dominated by definable ones as well, so that there the three
viewpoints --- sequences of flip operations, partition flips and definable flips --- are
interchangeable up to constant factors in the radius. This is how the statements of the
introduction, given for flip operations, relate to those of the body, given for definable flips. 

\subparagraph*{Combinatorial characterisations.}
We now formally define the combinatorial properties characterising monadic stability and
dependence, in their definable-flip formulation, the conversion coming from \Cref{fact:conversion}. The first property is the dense analogue of uniform
quasi-wideness \cite{NOdM11}. Note that the number $k$ of parameters depends only on the radius,
the target size $m$ being absorbed by the threshold $f(m)$; this pattern recurs.

\begin{definition}\label{def:flipflat}
A class $\C$ is \emph{flip-flat} if for every $r\in\N$ there are $k\in\N$ and $f:\N\to\N$ such that
for all $G\in\C$, $m\in\N$ and $A\subseteq V(G)$ with $|A|\ge f(m)$ there are $S\subseteq V(G)$
with $|S|\le k$ and $G'\in\flip(G/S)$ with $\ind^r_{G'}(A)\ge m$.
\end{definition}

\begin{theorem}[\cite{indiscernibles}]\label{thm:flipflat}
A class of graphs is flip-flat if and only if it is monadically stable.
\end{theorem}

Flip-flatness scatters a large set in one shot. The Flipper game instead localises repeatedly, and
asks how many rounds of flipping and zooming in are needed to exhaust the graph.

\begin{definition}\label{def:rank}
Fix $r\in\N$ and a graph $G$. For $U,S\subseteq V(G)$ the \emph{flip-separation rank of $U$ over
$S$} is
\[ \rk^G_r(U/S)=
   \begin{cases}
     0, & \text{if } U\subseteq S;\\
     1+\min_{s\in V(G)}\max_{u\in U}\rk^G_r\big(U\cap B^r_{S\cup\{s\}}(u)\,/\,S\cup\{s\}\big),
        & \text{otherwise,}
   \end{cases}
\]
and $\rk_r(G):=\rk^G_r(V(G)/\emptyset)$. A class $\C$ is \emph{winning for Flipper} if for every
$r\in\N$ there is $k\in\N$ with $\rk_r(G)\le k$ for every $G\in\C$.
\end{definition}

Equivalently, $\rk_r(G)$ is the value of a two-player game. The board consists of an \emph{arena},
initially $V(G)$, and the \emph{parameters} played so far, initially none. In each round Flipper
plays a vertex, adding it to the parameters, and Localiser answers with a vertex $u$ of the arena,
which then shrinks to its intersection with $B^r_S(u)$ over the enlarged parameter set $S$. Flipper
wins once nothing but parameters is left, and $\rk_r(G)$ is the number of rounds he needs under
optimal play. 
In the original form of the game \cite{flippergame}, sketched in the introduction, Flipper
plays a flip of the graph rather than a parameter and Localiser zooms into a ball of the flipped
graph; the form above is the \emph{confining Flipper game with quantifier-free definable
separation} of \cite[Section~3.1]{flippergame-full}, one of several forms shown there to be
equivalent up to a change of radius. This is the one our induction uses. 

\begin{theorem}[\cite{flippergame,flippergame-full}]\label{thm:finfliprank}
A class of graphs is winning for Flipper if and only if it is monadically stable.
\end{theorem}

The remaining two properties characterise monadic dependence. The first weakens flip-flatness by
asking only for a split into two far-apart parts; that it is a weakening is immediate, as an
$r$-independent set of size $2m$ splits into two such parts.

\begin{definition}\label{def:flipbreak}
A class $\C$ is \emph{flip-breakable} if for every $r\in\N$ there are $k\in\N$ and $f:\N\to\N$ such
that for all $G\in\C$, $m\in\N$ and $A\subseteq V(G)$ with $|A|\ge f(m)$ there are $S\subseteq V(G)$
with $|S|\le k$, $G'\in\flip(G/S)$ and $B_1,B_2\subseteq A$ of size at least $m$ with
$\dist_{G'}(b_1,b_2)>r$ for all $b_1\in B_1$, $b_2\in B_2$.
\end{definition}

\begin{theorem}[\cite{flipbreak}]\label{thm:flipbreak}
A class of graphs is flip-breakable if and only if it is monadically dependent.
\end{theorem}

The second replaces counting by weighing. Call a vertex \emph{$\varepsilon$-light} if it is not
$\varepsilon$-heavy.

\begin{definition}\label{def:flipsep}
A class $\C$ is \emph{flip-separable} if for every $r\in\N$ and $\varepsilon>0$ there is $k\in\N$
such that for every $G\in\C$ and every weighting $\w$ of $G$ there are $S\subseteq V(G)$ with
$|S|\le k$ and $G'\in\flip(G/S)$ such that $B^r_{G'}(u)$ is not $\varepsilon$-heavy for every
$\varepsilon$-light vertex $u$ of $G$.
\end{definition}

The restriction to $\varepsilon$-light vertices cannot be dropped, as an $\varepsilon$-heavy vertex
lies in its own ball, which no flip can lighten; it is vacuous for $\varepsilon$-balanced
weightings, under which every vertex is $\varepsilon$-light. Placing the balance on the objects
served rather than on the weighting is the convention of \cite{separability}, and we follow it
throughout. 

\begin{theorem}[\cite{separability}]\label{thm:flipsep}
A class of graphs is flip-separable if and only if it is monadically dependent.
\end{theorem}

The conclusion of \Cref{def:flipsep} may be restated in terms of scattered elements: if no
$r$-ball of $G'$ around an $\varepsilon$-light vertex is $\varepsilon$-heavy, then an
$\varepsilon$-heavy set $X$ of $\varepsilon$-light vertices cannot lie inside $B^r_{G'}(x)$ for any
$x\in X$, and so contains two elements at distance more than $r$. Thus flip-separability says that
a bounded flip pulls two elements apart in every $\varepsilon$-heavy set of $\varepsilon$-light
vertices, whereas flip-flatness pulls $m$ elements apart in every sufficiently large set;
flip-packability asks for both at once.

\section{Flip-packability}\label{sec:packing}

We now define the property formally, in the definable-flip formulation of \Cref{sec:prelims}. What we name is the family of parameter sets rather than the flips they
determine: it is what the proofs construct and what the algorithm of \Cref{sec:algorithmic}
computes, and each of its members stands for the boundedly many flips over it.

\begin{definition}[flip-nets]\label{def:flipnet}
Let $\w$ be a weighting of a graph $G$, let $r,m\in\N$ and $\varepsilon,\delta>0$. A family $\Fcal$
of subsets of $V(G)$ is an \emph{$(r,\varepsilon,\delta,m)$-flip-net for $\w$} if every
$\varepsilon$-heavy set $X\subseteq V(G)$ of $\delta$-light vertices satisfies
$\ind^r_{G'}(X)\ge m$ for some $F\in\Fcal$ and $G'\in\flip(G/F)$. It has \emph{order $(k,t)$} if
$\Fcal\subseteq[V(G)]^{\le k}$ and $|\Fcal|\le t$.
\end{definition}

\begin{definition}[flip-packability]\label{def:packing}
A class $\C$ of graphs is \emph{flip-packable} if for every $r\in\N$ there is $k\in\N$ such that for
every $\varepsilon>0$ and $m\in\N$ there are $\delta>0$ and $t\in\N$ such that every weighting of
every $G\in\C$ has an $(r,\varepsilon,\delta,m)$-flip-net of order $(k,t)$. 
\end{definition}
\Cref{thm:main-intro} of the introduction is the following statement with flip operations in place
of definable flips; the two are equivalent by the discussion following \Cref{fact:conversion}.

\begin{theorem}\label{thm:main}
A class of graphs is flip-packable if and only if it is monadically stable.
\end{theorem}

The proof comes in two stages. The first is an induction on the flip-separation rank, which
produces separation in the flip-metric; all the bookkeeping lives there, in two sequences
$\delta_k$ and $t_k$ recording how the balance and the size of the family degrade as the rank
grows. The second stage converts the flip-metric into a single flip. For graphs this could be done
by the metric conversion of \Cref{fact:conversion}(2), which keeps the margins polynomial and is
what \Cref{cor:polynomial} uses; here we use Ramsey's theorem instead, which keeps the parameter
sets at size $k$ and, having no relational analogue of \Cref{fact:conversion}(2) to rely on,
lets the proof carry over verbatim to relational structures (\Cref{app:relational}). We use
Ramsey's theorem in the following form.

\begin{fact}[Ramsey's theorem]\label{lem:ramsey}
There is a computable function $\Rcal:\N\times\N\to\N$ such that for all $m,c\in\N$, every
edge-colouring of the clique of size $\Rcal(m,c)$ with $c$ colours contains a monochromatic clique
of size $m$.
\end{fact}

\begin{lemma}\label{lem:winning-packing}
Every class of graphs that is winning for Flipper is flip-packable.
\end{lemma}

\begin{proof}
Fix $r\in\N$ throughout, and define two sequences by
\[ \delta_0(\varepsilon,m):=\tfrac{\varepsilon}{m},\qquad
   \delta_{k+1}(\varepsilon,m):=\delta_k\big(\tfrac{\varepsilon}{m},m\big), \]
\[ t_0(\varepsilon,m):=1,\qquad
   t_{k+1}(\varepsilon,m):=\big\lfloor\tfrac{m}{\varepsilon}\big\rfloor\cdot
   t_k\big(\tfrac{\varepsilon}{m},m\big)+1 \]
for $\varepsilon>0$ and $m\in\N$. Unrolling the recursions, the tolerance is divided by $m$ at each
level, so
\[ \delta_k(\varepsilon,m)=\frac{\varepsilon}{m^{k+1}}
   \qquad\text{and}\qquad
   t_k(\varepsilon,m)\;\le\;\prod_{i=1}^{k}\Big(\frac{m^i}{\varepsilon}+1\Big)
   \;\le\;\Big(\frac{2m^k}{\varepsilon}\Big)^{k}, \]
the last inequality for $\varepsilon\le1$, which we may assume, as for $\varepsilon>1$ there is no
heavy set to serve. As agreed in \Cref{sec:prelims}, weightings are normalised to total weight $1$
throughout the proof, so that $\varepsilon$-heavy means $\w(X)\ge\varepsilon$ and $\delta$-light
means $\w(v)<\delta$.

\medskip
\noindent\textbf{Stage 1: packing in the flip-metric.} We claim that for every $k\in\N$ the
following holds.
\begin{quote}
Let $\varepsilon>0$ and $m\in\N$, let $\w$ be a weighting of a graph $G$, and let
$A,S\subseteq V(G)$ satisfy $\rk_{3r}(A/S)\le k$. Then there is a family
$\Fcal\subseteq[V(G)]^{\le k}$ with $|\Fcal|\le t_k(\varepsilon,m)$ such that every
$\varepsilon$-heavy $X\subseteq A$ of $\delta_k(\varepsilon,m)$-light vertices satisfies
$\ind^r_{F\cup S}(X)\ge m$ for some $F\in\Fcal$.
\end{quote}
We argue by induction on $k$.

\emph{Base case.} If $k=0$ then $A\subseteq S$. Let $X\subseteq A$ be $\varepsilon$-heavy with
$\frac\varepsilon m$-light vertices. Then
\[ \varepsilon\;\le\;\w(X)\;<\;|X|\cdot\frac{\varepsilon}{m}, \]
so $X$ has more than $m$ elements, all of them in $S$. In the flip over $S$ that isolates every
vertex of $S$, these elements lie in distinct singleton components, hence pairwise at infinite
distance, and so $\ind^r_S(X)\ge m$. Thus $\Fcal:=\{\emptyset\}$ suffices.

\emph{Inductive step.} Assume the claim for $k$, and let $\varepsilon$, $m$, $\w$, $A$ and $S$ be as
in the statement for $k+1$. If $A\subseteq S$ we are done as in the base case, so assume not. Then
$\rk_{3r}(A/S)\le k+1$ unfolds, by \Cref{def:rank}, to
\[ \min_{s\in V(G)}\ \max_{u\in A}\ \rk_{3r}\big(A\cap B^{3r}_{S\cup\{s\}}(u)\,/\,S\cup\{s\}\big)
   \;\le\;k, \]
that is, Flipper has a move that drops the rank: there is $s\in V(G)$ such that, writing
$S':=S\cup\{s\}$,
\[ \rk_{3r}\big(A\cap B^{3r}_{S'}(u)\,/\,S'\big)\le k\qquad\text{for all }u\in A. \tag{$*$} \]
Fix such an $s$, and pack the graph with balls of large weight centred in $A$: let
\[ \Bcal:=\{B^r_{S'}(v_1),\dots,B^r_{S'}(v_\ell)\},\qquad v_1,\dots,v_\ell\in A, \]
be an inclusion-maximal family of pairwise disjoint balls of the $S'$-metric, centred in $A$, each
of weight at least $\frac{\varepsilon}{m}$. Such a family exists, and is finite, because
disjointness bounds the number of its members:
\[ \ell\cdot\frac{\varepsilon}{m}\;\le\;\sum_{i\in[\ell]}\w\big(B^r_{S'}(v_i)\big)\;\le\;\w(V(G))=1,
   \qquad\text{so}\qquad \ell\le\Big\lfloor\frac{m}{\varepsilon}\Big\rfloor. \]
This is what keeps the recursion from branching unboundedly; see \Cref{fig:packingstep}.

Each packed ball is a place where the induction hypothesis applies. For $i\in[\ell]$ put
\[ A_i:=A\cap B^{3r}_{S'}(v_i). \]
As $v_i\in A$, $(*)$ gives $\rk_{3r}(A_i/S')\le k$. The induction hypothesis, applied to $A_i$ and
$S'$ with tolerance $\frac\varepsilon m$ and the same $m$, therefore yields a family $\Fcal_i$ of at
most $t_k(\frac\varepsilon m,m)$ sets of size at most $k$ such that every $\frac\varepsilon m$-heavy
$Y\subseteq A_i$ of $\delta_k(\frac\varepsilon m,m)$-light vertices has $\ind^r_{F\cup S'}(Y)\ge m$
for some $F\in\Fcal_i$; note that
\[ \delta_k\big(\tfrac\varepsilon m,m\big)=\delta_{k+1}(\varepsilon,m) \]
by the definition of the sequence, so these are exactly the sets of vertices that the claim for
$k+1$ concerns. Collecting the families and adjoining $s$ to each member, we set
\[ \Fcal:=\{\{s\}\}\ \cup\ \bigcup_{i\in[\ell]}\big\{F\cup\{s\}:F\in\Fcal_i\big\}. \]
Every member of $\Fcal$ has size at most $k+1$, and
\[ |\Fcal|\;\le\;\ell\cdot t_k\big(\tfrac\varepsilon m,m\big)+1
   \;\le\;\Big\lfloor\frac m\varepsilon\Big\rfloor\cdot t_k\big(\tfrac\varepsilon m,m\big)+1
   \;=\;t_{k+1}(\varepsilon,m), \]
as required.

It remains to check that $\Fcal$ does what is asked. Fix an $\varepsilon$-heavy $X\subseteq A$ of
$\delta_{k+1}(\varepsilon,m)$-light vertices. There are two cases, according to whether the
parameters played so far already pull $X$ apart.

If $\ind^r_{S'}(X)\ge m$, then $F:=\{s\}\in\Fcal$ will do, since $F\cup S=S'$ and hence
$\ind^r_{F\cup S}(X)=\ind^r_{S'}(X)\ge m$.

Otherwise, let $Z$ be a maximal $r$-independent subset of $X$ in the $S'$-metric, so that $|Z|<m$.
By maximality, every $x\in X$ is at $S'$-distance at most $r$ from some $z\in Z$, that is, the
balls $B^r_{S'}(z)$ for $z\in Z$ cover $X$. Hence
\[ \varepsilon\;\le\;\w(X)\;\le\;\sum_{z\in Z}\w\big(X\cap B^r_{S'}(z)\big), \]
and as the sum has fewer than $m$ terms, some $x\in Z\subseteq X$ satisfies
\[ \w\big(X\cap B^r_{S'}(x)\big)\;\ge\;\frac{\varepsilon}{m}; \]
that is, $X$ concentrates in a single small ball. This ball is centred in $A$ and has weight at least
$\frac\varepsilon m$, so by the maximality of $\Bcal$ it is not disjoint from all the packed balls:
there is $i\in[\ell]$ with $B^r_{S'}(x)\cap B^r_{S'}(v_i)\ne\emptyset$. The triangle inequality then
places the small ball inside the corresponding larger one: for $y'$ in the intersection and any
$y\in B^r_{S'}(x)$,
\[ \dist_{S'}(v_i,y)\;\le\;\dist_{S'}(v_i,y')+\dist_{S'}(y',x)+\dist_{S'}(x,y)\;\le\;3r, \]
so $B^r_{S'}(x)\subseteq B^{3r}_{S'}(v_i)$. Consequently $X\cap B^r_{S'}(x)$ is a subset of
$A\cap B^{3r}_{S'}(v_i)=A_i$; it is $\frac\varepsilon m$-heavy, and its vertices, being vertices of
$X$, are $\delta_{k+1}(\varepsilon,m)$-light. The induction hypothesis for $i$ thus supplies
$F\in\Fcal_i$ with
\[ \ind^r_{F\cup S'}\big(X\cap B^r_{S'}(x)\big)\;\ge\;m. \]
Then $F^\star:=F\cup\{s\}$ lies in $\Fcal$, and since $F^\star\cup S=F\cup S'$ and
$X\supseteq X\cap B^r_{S'}(x)$, we get $\ind^r_{F^\star\cup S}(X)\ge m$. This completes the
induction.

\medskip
\noindent\textbf{Stage 2: from the flip-metric to a single flip.} Let $k$ be such that
$\rk_{3r}(G)\le k$ for all $G\in\C$, which exists as $\C$ is winning for Flipper, and given
$\varepsilon>0$ and $m\in\N$ put
\[ m':=\Rcal\big(m,\Xi(k)\big),\qquad \delta:=\delta_k(\varepsilon,m'),\qquad t:=t_k(\varepsilon,m'). \]
Let $\w$ be a weighting of some $G\in\C$. Since $\rk_{3r}(V(G)/\emptyset)\le k$, Stage~1 applied to
$A:=V(G)$ and $S:=\emptyset$, with tolerance $\varepsilon$ and target $m'$, provides a family
$\Fcal\subseteq[V(G)]^{\le k}$ of size at most $t$ such that every $\varepsilon$-heavy $X$ of
$\delta$-light vertices admits $F\in\Fcal$ with
\[ \ind^r_F(X)\;\ge\;m'. \]
We show that $\Fcal$ is an $(r,\varepsilon,\delta,m)$-flip-net of order $(k,t)$ for $\w$.

Fix such an $X$ and $F$, and let $Y\subseteq X$ be a set of $m'$ elements pairwise at
flip-distance greater than $r$ over $F$. Since $\dist_F$ is the maximum of $\dist_{G'}$ over
$G'\in\flip(G/F)$, for each pair $u,v\in Y$ there is some $G'\in\flip(G/F)$ with
$\dist_{G'}(u,v)>r$ --- but possibly a different one for each pair. Colour the pair $uv$ by one such
$G'$. The number of colours is at most $|\flip(G/F)|\le\Xi(|F|)\le\Xi(k)$, so by the choice of $m'$
and \Cref{lem:ramsey} there is $Y'\subseteq Y$ of size $m$ all of whose pairs receive the same
colour $G'$. Then every two elements of $Y'$ are at distance greater than $r$ in this single
$G'\in\flip(G/F)$, so $Y'$ is $r$-independent in $G'$, and as $Y'\subseteq X$ we conclude
$\ind^r_{G'}(X)\ge m$, as required.
\end{proof}

\begin{figure}[t]
\centering
\includegraphics[width=0.5\textwidth]{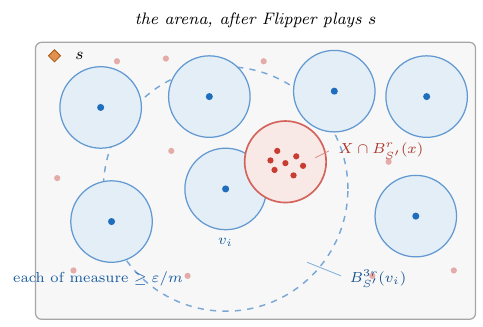}
\caption{\small The inductive step of \Cref{lem:winning-packing}. The shaded discs are a maximal
family of pairwise disjoint $r$-balls over $S'=S\cup\{s\}$ of weight at least $\varepsilon/m$; at
most $m/\varepsilon$ of them fit. A heavy $X$ that is not already pulled apart by $S'$ concentrates
in a single $r$-ball, drawn in red, which by maximality meets one of the packed balls, and the
triangle inequality then places it inside $B^{3r}_{S'}(v_i)$, where the induction hypothesis
applies.}
\label{fig:packingstep}
\end{figure}

The converse direction of \Cref{thm:main} is the case $\varepsilon=1$.

\begin{lemma}\label{lem:packing-flat}
Every flip-packable class of graphs is flip-flat.
\end{lemma}

\begin{proof}
Given $r$, let $k$ be as in \Cref{def:packing}, and given $m$ let $\delta:=\delta(r,1,m)$ be the
balance of \Cref{def:packing} for tolerance $\varepsilon=1$; put $f(m):=\lceil1/\delta\rceil+1$.
Let $G\in\C$ and $A\subseteq V(G)$ with $|A|\ge f(m)$, and consider the indicator weighting $\w_A$.
Its total weight is $|A|>1/\delta$, so every vertex is $\delta$-light, and $A$ itself is $1$-heavy.
An $(r,1,\delta,m)$-flip-net $\Fcal$ of order $(k,t)$ for $\w_A$ therefore serves $A$: some
$F\in\Fcal$ of size at most $k$ and some $G'\in\flip(G/F)$ have $\ind^r_{G'}(A)\ge m$, as
flip-flatness requires.
\end{proof}

\begin{proof}[Proof of \Cref{thm:main}]
If $\C$ is monadically stable then it is winning for Flipper by \Cref{thm:finfliprank}, hence
flip-packable by \Cref{lem:winning-packing}. The converse is \Cref{lem:packing-flat} together with
\Cref{thm:flipflat}.
\end{proof}

The proof of \Cref{lem:packing-flat} uses no weighting other than an indicator one, so
\Cref{thm:main} remains true if \Cref{def:packing} is restricted to indicator weightings, for which
the balance is automatic once the set carrying the weight has more than $1/\delta$ vertices; the
definition then takes the following purely combinatorial form.

\begin{corollary}[counting form]\label{cor:counting}
A class $\C$ of graphs is monadically stable if and only if for every $r\in\N$ there is $k\in\N$
such that for every $\varepsilon>0$ and $m\in\N$ there are $n_0,t\in\N$ with the following
property: for every $G\in\C$ and every $A\subseteq V(G)$ with $|A|\ge n_0$ there is a family
$\Fcal\subseteq[V(G)]^{\le k}$ with $|\Fcal|\le t$ such that every $X\subseteq A$ with
$|X|\ge\varepsilon|A|$ satisfies $\ind^r_{G'}(X)\ge m$ for some $F\in\Fcal$ and $G'\in\flip(G/F)$.
\end{corollary}

\section{Flip-separability and $2$-flip-packability}\label{sec:backtosep}

We now examine the case $m=2$ of flip-packability.

\begin{definition}\label{def:2packing}
A class $\C$ of graphs is \emph{$2$-flip-packable} if for every $r\in\N$ there is $k\in\N$ such
that for every $\varepsilon>0$ there are $\delta>0$ and $t\in\N$ such that every weighting of every
$G\in\C$ has an $(r,\varepsilon,\delta,2)$-flip-net of order $(k,t)$.
\end{definition}

This is \Cref{def:packing} with $m$ fixed to $2$, so every flip-packable class is $2$-flip-packable.
The main result of this section is that it characterises monadic dependence.

\begin{theorem}\label{thm:2packing-equiv}
A class of graphs is $2$-flip-packable if and only if it is monadically dependent.
\end{theorem}

We prove this by sandwiching $2$-flip-packability between monadic dependence and
flip-separability. That monadic dependence implies $2$-flip-packability follows by analysing the proof of
\cite{separability}: the family of parameter sets built there already has the two-tier shape of
\Cref{def:2packing}, and is merged into a single flip at the very end. That
$2$-flip-packability implies flip-separability is the merging itself, done with the metric
conversion of \Cref{fact:conversion}(2); flip-separability implies monadic dependence by
\Cref{thm:flipsep}. We take the two steps in this order.

\subparagraph*{From monadic dependence to $2$-flip-packability.}
The argument of \cite{separability} constructs, for a balanced weighting, a bounded family of
bounded parameter sets over which every ball of the combined flip-metric is light.

\begin{fact}[follows from {\cite[Lemma~20]{separability}}]\label{fact:sparsifying-family}
Let $\C$ be a monadically dependent class of graphs. For every $r\in\N$ there is $k\in\N$, and for
every $\varepsilon>0$ there is $t\in\N$, such that for every $G\in\C$ and every
$\varepsilon$-balanced weighting $\w$ of $G$ there is a family $\Fcal$ of at most $t$ subsets of
$V(G)$, each of size at most $k$, with
\[ \w\big(B^r_{\Fcal}(v)\big)\;\le\;\varepsilon\cdot\w(V(G))\quad\text{for every }v\in V(G),
   \qquad\text{where }B^r_{\Fcal}(v):=\bigcap_{S\in\Fcal}B^r_S(v). \]
\end{fact}

In \cite[Lemma~20]{separability} both the size of the sets and their number are stated with a
dependence on $\varepsilon$, and the letters $k$ and $t$ play the opposite roles. That the size may
be taken to depend on $r$ alone is visible in the proof: the sets have size $t(\C,3r')$, the
flip-breakability parameter of \cite[Corollary~17]{separability} at radius $3r'$, where $r'=r'(r)$
is the locality radius of \cite[Lemma~14]{separability} for formulas of quantifier rank $r$; the
initial family consists of sets of that size, and the compression step
\cite[Lemma~19]{separability} exchanges such sets for sets of the same size. Only the number of sets
depends on $\varepsilon$.

\begin{lemma}\label{lem:dependent-packing}
Every monadically dependent class of graphs is $2$-flip-packable. More precisely, for every
$r\in\N$ there is $k\in\N$ such that for every $\varepsilon>0$ there is $t\in\N$ such that every
weighting of every $G\in\C$ has an $(r,\varepsilon,\varepsilon^2/2,2)$-flip-net of order $(k,t)$.
\end{lemma}

\begin{proof}
Fix $r$, let $k$ be as in \Cref{fact:sparsifying-family} for $r$, and given $\varepsilon>0$ let
$t$ be as in that fact for tolerance $\varepsilon/2$. Put $\delta:=\varepsilon^2/2$. Let $\w$ be a
weighting of some $G\in\C$, normalised to total weight $1$, and let $H$ be the set of
$\delta$-heavy vertices. If $\w(V(G)\setminus H)<\varepsilon$ there is no $\varepsilon$-heavy set of
$\delta$-light vertices, and any family is a flip-net; so assume $\w(V(G)\setminus H)\ge\varepsilon$.

Let $\w'$ be the restriction of $\w$ to $V(G)\setminus H$, that is, $\w'(v):=\w(v)$ for $v\notin H$
and $\w'(v):=0$ for $v\in H$. Every vertex has $\w'(v)<\delta=\tfrac\varepsilon2\cdot\varepsilon
\le\tfrac\varepsilon2\cdot\w'(V(G))$, so $\w'$ is $\frac\varepsilon2$-balanced, and
\Cref{fact:sparsifying-family} applied to $\w'$ with tolerance $\frac\varepsilon2$ yields a family
$\Fcal$ of at most $t$ sets of size at most $k$ with
\[ \w'\big(B^r_{\Fcal}(v)\big)\;\le\;\tfrac\varepsilon2\cdot\w'(V(G))\;<\;\varepsilon
   \qquad\text{for every }v\in V(G). \]

We claim that $\Fcal$ is an $(r,\varepsilon,\delta,2)$-flip-net for $\w$. Let $X$ be
$\varepsilon$-heavy for $\w$ with $\delta$-light vertices. Then $X\subseteq V(G)\setminus H$, so
$\w'(X)=\w(X)\ge\varepsilon$. Pick $u\in X$. As $\w'(B^r_{\Fcal}(u))<\varepsilon\le\w'(X)$, some
$v\in X$ lies outside $B^r_{\Fcal}(u)$, that is, $\dist_S(u,v)>r$ for some $S\in\Fcal$. Since
$\dist_S$ is the maximum of $\dist_{G'}$ over $G'\in\flip(G/S)$, some $G'\in\flip(G/S)$ has
$\dist_{G'}(u,v)>r$, and so $\ind^r_{G'}(X)\ge2$.
\end{proof}

\subparagraph*{From $2$-flip-packability to flip-separability.}
The first step extracts from $2$-flip-packability a bounded parameter set over which every ball of
the flip-metric around a light vertex is light. The heavy vertices are few, at most $1/\delta$ of
them, and are simply played as parameters.

\begin{lemma}[light balls]\label{lem:light-balls}
Let $\C$ be $2$-flip-packable. For every $r\in\N$ and $\varepsilon>0$ there are $k\in\N$ and
$\delta>0$ such that for every $G\in\C$ and every weighting $\w$ of $G$ there is $S\subseteq V(G)$
with $|S|\le k$ such that $B^r_S(u)$ is not $\varepsilon$-heavy for every $\delta$-light vertex $u$.
\end{lemma}

\begin{proof}
Apply \Cref{def:2packing} with radius $2r$ and tolerance $\varepsilon$, obtaining $k(2r)$, $\delta$
and $t$, and put $k:=k(2r)\cdot t+\lfloor1/\delta\rfloor$. Given $G$ and $\w$, let $\Fcal$ be a
$(2r,\varepsilon,\delta,2)$-flip-net of order $(k(2r),t)$ for $\w$, let $H$ be the set of
$\delta$-heavy vertices, of which there are at most $1/\delta$, and put $S:=\bigcup\Fcal\cup H$,
so that $|S|\le k$.

Let $u$ be $\delta$-light. Every vertex of $H\subseteq S$ is isolated in some $S$-flip, hence at
infinite flip-distance from $u$, so $B^r_S(u)$ consists of $\delta$-light vertices. Suppose it were
$\varepsilon$-heavy. Then the flip-net serves it: there are $F\in\Fcal$ and $G'\in\flip(G/F)$ with
$\ind^{2r}_{G'}(B^r_S(u))\ge2$, that is, $a,b\in B^r_S(u)$ with $\dist_{G'}(a,b)>2r$. Hence
$\dist_F(a,b)>2r$, and $\dist_S(a,b)>2r$ as $F\subseteq S$. But $a,b\in B^r_S(u)$ gives
$\dist_S(a,b)\le2r$ by the triangle inequality, a contradiction.
\end{proof}

From here one could reach flip-breakability, by splitting a large set along the light balls and
merging the flips by a bipartite Ramsey argument; this is done uniformly for all the variants of
\Cref{sec:variants} in \Cref{app:variants}. Here we pass to flip-separability directly, through the
metric conversion of \Cref{fact:conversion}(2).

\begin{lemma}\label{lem:packing-sep}
Every $2$-flip-packable class of graphs is flip-separable.
\end{lemma}

\begin{proof}
Let $\C$ be $2$-flip-packable. Fix $r\in\N$ and $\varepsilon>0$, and let $k_1$ and $\delta_1$ be
the constants of \Cref{lem:light-balls} for radius $6r$ and tolerance $\varepsilon$. Let $\w$ be a
weighting of some $G\in\C$, and let $S\subseteq V(G)$ with $|S|\le k_1$ be as in that lemma.

We claim that $B^{6r}_S(u)$ is not $\varepsilon$-heavy for every $\varepsilon$-light $u$. If $u$ is
$\delta_1$-light this is the lemma. Otherwise $u$ is $\delta_1$-heavy, hence belongs to $S$ and is
isolated in some $S$-flip, so $B^{6r}_S(u)=\{u\}$, which is not $\varepsilon$-heavy as $u$ is
$\varepsilon$-light.

Now \Cref{fact:conversion}(2) gives $T\subseteq V(G)$ with $|T|\le k_1+(2^{k_1}+k_1)^2$, a bound
depending on $r$ and $\varepsilon$ alone, and $G'\in\flip(G/T)$ with
$B^r_{G'}(u)\subseteq B^{6r}_S(u)$ for all $u$. So $B^r_{G'}(u)$ is not $\varepsilon$-heavy for
every $\varepsilon$-light $u$, which is what \Cref{def:flipsep} asks.
\end{proof}

\begin{proof}[Proof of \Cref{thm:2packing-equiv}]
A monadically dependent class is $2$-flip-packable by \Cref{lem:dependent-packing}. Conversely, a
$2$-flip-packable class is flip-separable by \Cref{lem:packing-sep}, hence monadically dependent by
\Cref{thm:flipsep}.
\end{proof}

By \Cref{thm:flipsep,thm:flipbreak}, $2$-flip-packability is therefore equivalent to
flip-separability and to flip-breakability alike. The proof passes through the former because
$2$-flip-packability is its two-tier form; the direct passage to the latter is \Cref{lem:app-m2-break}.

\section{Everything is packing}\label{sec:variants}

We now make the eight statements of \Cref{thm:everything} precise, and explain why each of them
characterises its entry in \Cref{tab:variants}; the formal definitions
(\Cref{def:app-packing,def:app-flat-break}) and the proofs are in \Cref{app:variants} of the appendix. Throughout,
$\rho$ ranges over $\N\cup\{\infty\}$, and at radius $\infty$ balls are connected components.

For $F,X\subseteq V(G)$ let $\ind^\rho_{G-F}(X)$ be the largest size of a subset of $X\setminus F$
whose elements are pairwise at distance greater than $\rho$ in $G-F$, which for $\rho=\infty$ means
lying in pairwise distinct components of $G-F$; likewise $\ind^\infty_{G'}(X)$ counts elements of
$X$ in distinct components of $G'$. \emph{Deletion-packability at radius $\rho$} is the statement
obtained from \Cref{def:packing} by replacing the conclusion $\ind^r_{G'}(X)\ge m$ of the flip-net
with $\ind^\rho_{G-F}(X)\ge m$, and \emph{flip-packability at radius $\infty$} is obtained by
taking $r=\infty$ throughout; flip-packability itself is flip-packability at every finite radius.
In each case \emph{$2$-packability} is the property with $m$ fixed to $2$, as in
\Cref{def:2packing}; the size $k$ of each parameter set is still required to depend on the radius
alone. Two of the entries are already proved, namely that flip-packability characterises monadic
stability (\Cref{thm:main}), and that $2$-flip-packability characterises monadic dependence
(\Cref{thm:2packing-equiv}).

\begin{table}[h]
\caption{\Cref{tab:variants-intro}, repeated here for convenience. The column $m=2$ records $2$-packability; by
\Cref{rem:bounded-m} the same entries arise for any fixed $m$, and only unbounded $m$ gives the
left-hand column. No closure assumptions on $\C$ are made.}
    \label{tab:variants}
    \centering
    \begin{tabular}{llll}
        \toprule
        \textbf{operation} & \textbf{radius} & \textbf{all $m$} & \textbf{$m=2$} \\
        \midrule
        deletions & $r\in\N$ & nowhere dense & nowhere dense \\
        flips & $r\in\N$ & monadically stable & monadically dependent \\
        deletions & $\infty$ & bounded treedepth & bounded treewidth \\
        flips & $\infty$ & bounded shrubdepth & bounded cliquewidth \\
        \bottomrule
    \end{tabular}
\end{table}

Each row is proved as a cycle (\Cref{fig:variants-map}). The packing property yields the
corresponding variant of flatness in the left-hand column, respectively of breakability in the
right-hand one; these are the properties classified in \cite[Section~18]{flipbreak-full}, which
supplies the characterisations of the eight entries; and a return arrow, different in the two
columns, reconstructs the packing property from the entry. The forward arrows are the arguments of
\Cref{sec:packing,sec:backtosep} with the operation and the radius changed; the return arrows are
where the work lies.

\subparagraph*{The all-$m$ column.}
The forward arrow is \Cref{lem:packing-flat}: with $\varepsilon=1$ and the indicator weighting of a
large set $A$, the flip-net serves $A$ itself, which is the corresponding variant of flatness ---
uniform quasi-wideness in the first row, flip-flatness in the second, and in the last two the
statement that every large set has $m$ elements in distinct components after boundedly many
deletions, respectively after a bounded flip.

The return arrow is the induction of \Cref{lem:winning-packing}, run on a depth parameter. For
deletions at finite radius this is the rank of the Splitter game of \cite{deciding}, bounded exactly
on nowhere dense classes; for flips at finite radius it is the flip-separation rank of
\Cref{def:rank}. At radius $\infty$, writing $\rk^\del_\infty$ for the deletion rank and $C_{S'}(u)$
for the component of $u$ in $G-S'$, the recursion
\[ \rk^\del_\infty(U/S)=1+\min_s\max_u\rk^\del_\infty\big(U\cap C_{S\cup\{s\}}(u)\,/\,S\cup\{s\}\big) \]
is the elimination-forest recursion $\td(G)=1+\min_v\max_C\td(C)$ on connected arenas, so
$\rk^\del_\infty$ is treedepth up to an additive one; and $\rk_\infty$, the flip-separation rank at
radius $\infty$, is bounded exactly on classes of bounded shrubdepth (see the paragraph after \Cref{fact:app-depth}). The induction uses only three
things: that the rank drops when a vertex is played, that disjoint balls of weight at least
$\frac\varepsilon m$ are at most $\frac m\varepsilon$ in number, and that a ball meeting a packed
ball lies inside a slightly larger one. All three survive verbatim for deletions and at radius
$\infty$, where the last is the observation that two intersecting components coincide. For
deletions the Ramsey stage disappears, as there is no family of flips to merge, but a deleted vertex
is not counted whereas an isolated one is, which costs two small adjustments recorded in
\Cref{lem:app-rank-packing}.

\subparagraph*{The $m=2$ column.}
The forward arrow starts with \Cref{lem:light-balls}: from $2$-packability one obtains, for every
weighting, a bounded set $S$ over which every ball of the associated metric around a light vertex
is light. Applied to the indicator weighting of a large set $W$, this splits $W$ into two large
parts lying far apart, by the argument of \cite[Lemma~24]{separability}; at radius $\infty$ the
split is simpler still, as every component meets $W$ in at most half of it. For deletions the parts
are far apart in $G-S$ and there is nothing more to do; for flips they are far apart in the
flip-metric, and a bipartite Ramsey argument over the at most $\Xi(|S|)$ flips over $S$ places them
in a single flip. What results is the corresponding variant of breakability
(\Cref{lem:app-m2-break}).

The return arrows are where the four rows part ways. The first follows from its own left-hand
entry, nowhere density implying the packing property outright. The second is
\Cref{lem:dependent-packing}: monadic dependence yields $2$-packability by compression of a
sparsifying family, rather than by any depth parameter. The third and fourth are the list of
centroids of the introduction, with bags of a tree decomposition, respectively bounded flips, in
place of centroids. Both rows come with a decomposition supplying, for \emph{any} given weighting,
one small object after which every component carries at most a constant fraction of the weight, and
one lists the object balanced for every region heavy enough to hold $X$, at every level of a
recursion of depth $\Ocal(\log(1/\varepsilon))$ and branching $\Ocal(1/\varepsilon)$. The family is
large but bounded in terms of $\varepsilon$, while each of its members is a single bag or flip, of
size independent of $\varepsilon$; that asymmetry is what \Cref{def:2packing} asks for, and what
following $X$ cannot give, since the separator that follows $X$ grows with $\log(1/\varepsilon)$
(\Cref{lem:app-halving}).

Placing the two return arrows side by side is what explains the table. Pulling two elements apart
can always be arranged by recursing into the heavy part and paying for the recursion in the
number of separators, which is available whenever balanced separators are; pulling $m$ apart with
separators whose size may grow with neither $m$ nor $\varepsilon$ requires the decomposition itself
to have bounded depth, which is what the left-hand entries have and the right-hand entries lack.

\begin{figure}[ht]
\centering
\begin{tikzpicture}[
  bx/.style={rounded corners=2.5pt,draw,line width=0.6pt,align=center,
             font=\footnotesize,inner sep=3.4pt,minimum height=7.2mm},
  pack/.style={bx,draw=pkc!70,fill=pkc!10,minimum width=2.95cm},
  flat/.style={bx,draw=ftc!70,fill=ftc!10},
  brk/.style={bx,draw=brc!70,fill=brc!10},
  not/.style={bx,draw=ntc!60,fill=ntc!8},
  ar/.style={-{Stealth[length=4pt,width=3.4pt]},line width=0.65pt,black!65},
  lb/.style={font=\scriptsize,inner sep=1.6pt,fill=white,align=center}]

\node[pack] (P)  at (0,1.15) {packability};
\node[flat] (F)  at (4.9,1.15) {flatness};
\node[not]  (N1) at (9.9,1.15) {nowhere dense / mon.\ stable\\ treedepth / shrubdepth};
\draw[ar] (P) -- node[lb,above]{$\varepsilon=1$,\\ indicator weighting} (F);
\draw[ar] (F) -- node[lb,above]{known} (N1);
\draw[ar] (N1.south) .. controls +(0,-0.85) and +(0,-0.85) ..
   node[lb,below]{induction on the depth parameter} ($(P.south east)+(-0.45,0)$);

\node[pack] (Q)  at (0,-2.05) {$2$-packability};
\node[brk]  (B)  at (4.9,-2.05) {breakability};
\node[not]  (N2) at (9.9,-2.05) {nowhere dense / mon.\ dependent\\ treewidth / cliquewidth};
\draw[ar] (Q) -- node[lb,above]{light balls,\\ then split} (B);
\draw[ar] (B) -- node[lb,above]{known} (N2);
\draw[ar] (N2.south) .. controls +(0,-0.85) and +(0,-0.85) ..
   node[lb,below]{compression, or iterated halving} (Q.south);

\draw[ar,black!45] (P.south) -- node[lb,right,pos=0.5]{fix $m=2$} (Q.north);

\end{tikzpicture}
\caption{\small How the entries of \Cref{tab:variants} are obtained. Each row is a cycle: the
packing property yields the corresponding variant of flatness, respectively breakability, and these
characterise the notions on the right by \cite[Section~18]{flipbreak-full}. The return arrows are
the converse constructions.}
\label{fig:variants-map}
\end{figure}
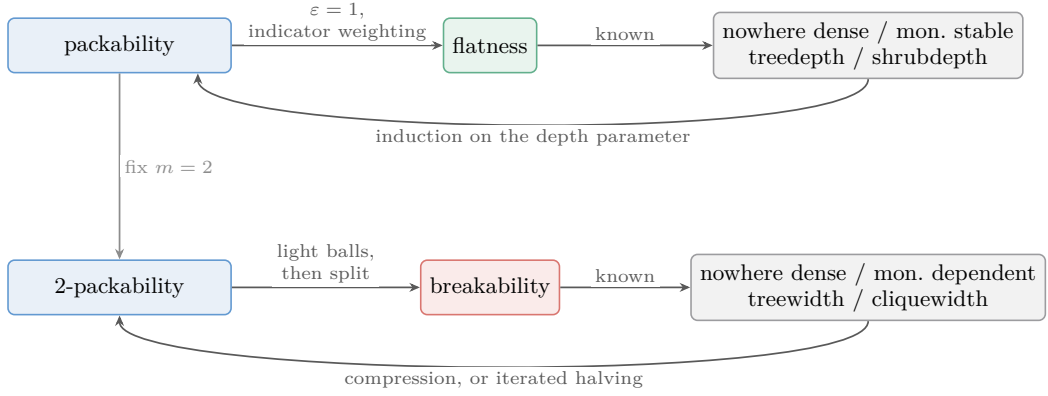

\begin{remark}\label{rem:bounded-m}
The left-hand column of \Cref{tab:variants} is really the case of unbounded $m$: fixing any other
value of $m$ changes nothing in any of the four rows. This is not obtained by merging the family of
the $m=2$ case into one set over which every ball is light and choosing $m$ elements of a heavy set
greedily, as that set has size $k\cdot t$, which depends on $\varepsilon$. Instead one iterates the
$m=2$ case $m-1$ times: two far-apart elements of a heavy set $X$ cut it into three parts, the two
balls around them and the remainder, one of which is heavy; one recurses into that part with the
tolerance divided by three, and the element on the other side of the cut is, by the triangle
inequality, far from everything in it. Unions of $m-1$ members of the families produced along the
way yield $m$ elements pairwise far apart in the metric over a set of size $(m-1)k$, with $k$ the
parameter of the $m=2$ case, which for flips \Cref{fact:conversion}(2) turns into a single flip at
a constant factor in the radius. Only the tolerance and the size of the family degrade with
$\varepsilon$, while the size of the parameter sets depends on $m$ and the radius alone, as
\Cref{def:packing} permits once $m$ is fixed. Nothing beyond the triangle inequality is used, so
the same holds for deletions and at radius $\infty$, where balls are components. Thus the two
columns are separated not by two against many but by bounded against unbounded $m$.
\end{remark}

\section{Computing flip-nets}\label{sec:algorithmic}

The proof of \Cref{lem:winning-packing} is effective once Flipper's moves are: the packing step is
a greedy scan, the flip-metric over a bounded parameter set is computed by enumerating the at most
$\Xi(k)$ flips over it, and the recursion tree has bounded size. What is needed is a winning
strategy for Flipper, in the confining form of \Cref{def:rank}, that can be computed efficiently.
We obtain one by simulating the algorithmic strategy of Gajarsk\'y et al.\
\cite[Theorem~11.2]{flippergame-full}, which is stated for the original form of the game in which
Flipper plays flips rather than vertices, using the metric conversion of \Cref{fact:conversion}(1)
to turn each flip into boundedly many parameters.

\begin{theorem}\label{thm:rank-algo}
Let $\C$ be a monadically stable class of graphs and $r\in\N$. There is $k\in\N$ such that
$\rk_r(G)\le k$ for every $G\in\C$, and Flipper's moves in the game of \Cref{def:rank} at radius $r$
can be computed in time $\Ocal_{\C,r}(n^2)$ on an $n$-vertex graph $G\in\C$.
\end{theorem}

\begin{theorem}\label{thm:algo}
Let $\C$ be a monadically stable class of graphs, let $r$, $\varepsilon$ and $m$ be given, and let
$k$, $t$ and $\delta$ be the constants that the proof of \Cref{thm:main} provides for them. There
is an algorithm that, given an $n$-vertex graph $G\in\C$ and a weighting
$\w:V(G)\to\mathbb Q_{\ge0}$, computes an $(r,\varepsilon,\delta,m)$-flip-net of order $(k,t)$ for
$\w$ in time $\Ocal_{\C,r,\varepsilon,m}(n^3)$, counting arithmetic operations on the weights at
unit cost.
\end{theorem}

The proofs are in \Cref{app:algorithmic}. Two remarks on the statement. 
First, the
weighting need not be given explicitly. The algorithm accesses $\w$ only through the weights of
balls of the flip-metric, so an oracle answering queries $\w(X)$ for $X\subseteq V(G)$ suffices, and
it is queried $\Ocal(n)$ times per node of the recursion, hence $\Ocal_{\C,r,\varepsilon,m}(n)$
times in total. Second, \Cref{thm:algo-intro} of the introduction follows: each member $F$ of the
net determines at most $\Xi(k)$ flips, each computable in time $\Ocal(n^2)$ and obtained from $G$
by at most $(2^k+k)^2$ flip operations, so the family of flipped graphs is computed from the net in
time $\Ocal(t\,\Xi(k)\,n^2)$, and a heavy query set is then served by scanning boundedly many
explicit graphs for one in which it is $r$-scattered.

\subparagraph*{Dependent classes.}
The Flipper game is not available on monadically dependent classes, and the family of
\Cref{lem:dependent-packing} comes instead from the compression argument of
\cite[Lemma~20]{separability}. That argument discards sets through a colouring supplied by Gaifman's
locality theorem, and computing that colouring amounts to computing local first-order types, which
is essentially the model-checking problem on $\C$. The colouring, however, is needed only in the
analysis: choosing the sets to keep greedily, and using locality to bound the number of greedy steps
rather than to make them, makes the compression effective. Its only external ingredient is then
flip-breakability in the algorithmic form of \cite[Theorem~13.2]{flipbreak-full}.

\begin{theorem}\label{thm:algo-dependent}
Let $\C$ be a monadically dependent class of graphs, let $r\in\N$ and $\varepsilon>0$, and let $k$
and $t$ be the constants of \Cref{lem:dependent-packing} for them. There is an algorithm that,
given an $n$-vertex graph $G\in\C$ and a weighting $\w:V(G)\to\mathbb Q_{\ge0}$, computes an
$(r,\varepsilon,\varepsilon^2/2,2)$-flip-net of order $(k,t)$ for $\w$ in time
$\Ocal_{\C,r,\varepsilon}(n^5)$.
\end{theorem}

By verifying that \Cref{lem:light-balls} and \Cref{lem:packing-sep} can be carried out effectively, we thus obtain an algorithmic form of
\Cref{thm:flipsep}, which \cite{separability} only states existentially.

\begin{corollary}\label{cor:algo-sep}
Let $\C$ be a monadically dependent class of graphs, $r\in\N$ and $\varepsilon>0$. There is $k\in\N$
and an algorithm that, given an $n$-vertex graph $G\in\C$ and a weighting
$\w:V(G)\to\mathbb Q_{\ge0}$, computes a set $T\subseteq V(G)$ with $|T|\le k$ and a flip
$G'\in\flip(G/T)$ such that $B^r_{G'}(u)$ is not $\varepsilon$-heavy for every $\varepsilon$-light
vertex $u$. It runs in time $\Ocal_{\C,r,\varepsilon}(n^5)$, and in time
$\Ocal_{\C,r,\varepsilon}(n^3)$ if $\C$ is monadically stable.
\end{corollary}

The proofs are in \Cref{app:dependent}. The difference in the running times between monadically stable and dependent classes is precisely due the difference in the running times of \Cref{thm:algo} and \Cref{thm:algo-dependent}.

\subparagraph*{Polynomial margins.}
The \emph{margins} of uniform quasi-wideness --- how large a set must be to contain $m$ elements
that are $r$-independent after boundedly many deletions --- were originally obtained by iterated
Ramsey arguments and are non-elementary \cite{NOdM11}. Kreutzer, Rabinovich and Siebertz \cite{polynomialUQW}
showed, through the indiscernible sequences of Malliaris and Shelah \cite{MS14}, that for every fixed $r$ they
can be taken polynomial in $m$, and that the sets can be computed in polynomial time. 
The dense counterpart, flip-flatness with polynomial
margins and a cubic-time algorithm, is \cite[Theorem~1.4]{indiscernibles}, obtained along the
same route. The proof of \Cref{lem:winning-packing} gives a second route to it, through the Flipper
game, and extends it to the packing property: the sequences of \Cref{lem:winning-packing} are
$\delta_k(\varepsilon,m)=\varepsilon/m^{k+1}$ and $t_k(\varepsilon,m)\le(2m^k/\varepsilon)^k$, so
$1/\delta_k$ and $t_k$ are polynomial in $m$ and $1/\varepsilon$ with exponents depending on the
rank alone, and the only step of the proof that is not polynomial is Ramsey's theorem, which for
graphs \Cref{fact:conversion}(2) replaces at the cost of parameter sets of size exponential in the
rank.

\begin{corollary}[polynomial margins]\label{cor:polynomial}
Let $\C$ be a monadically stable class of graphs, let $r\in\N$, let $k$ bound $\rk_{18r}$ on
$\C$, and put $s:=k+(2^k+k)^2$. For every $\varepsilon>0$ and $m\in\N$, every weighting $\w$ of
every $G\in\C$ has an $(r,\varepsilon,\varepsilon/m^{k+1},m)$-flip-net of order
$\big(s,(2m^k/\varepsilon)^k\big)$, computable from $G$ and $\w$ in time
$\Ocal_{\C,r,\varepsilon,m}(n^3)$.
\end{corollary}

\begin{proof}
Stage~1 of the proof of \Cref{lem:winning-packing} at radius $6r$ yields
$\Fcal\subseteq[V(G)]^{\le k}$ with $|\Fcal|\le t_k(\varepsilon,m)\le(2m^k/\varepsilon)^k$ such
that every $\varepsilon$-heavy $X$ of $\delta_k(\varepsilon,m)$-light vertices has
$\ind^{6r}_F(X)\ge m$ for some $F\in\Fcal$, where $\delta_k(\varepsilon,m)=\varepsilon/m^{k+1}$.
For each $F\in\Fcal$, \Cref{fact:conversion}(2) gives $T_F\supseteq F$ with $|T_F|\le s$ and
$G'_F\in\flip(G/T_F)$ with $B^r_{G'_F}\subseteq B^{6r}_F$, so that $\ind^r_{G'_F}(X)\ge m$; the
family $\{T_F:F\in\Fcal\}$ is the required net. It is computed by the algorithm of
\Cref{thm:algo} with $m$ in place of $m'$, followed by the conversion, whose proof in
\cite{separability} is constructive.
\end{proof}

At $\varepsilon=1$ this is flip-flatness with margin $m^{k+1}$ and parameter sets of size $s$,
that is, \cite[Theorem~1.4]{indiscernibles} with the exponent made explicit in terms of the rank;
and since Stage~1 is constructive on a given set --- to serve a specific $X$ one need not pack,
but simply descends into the ball in which $X$ concentrates, with Flipper's move from
\Cref{thm:rank-algo} at each level --- the cubic-time algorithm of that theorem is recovered as
well. What is new is the $\varepsilon$ axis: the number $t$ of parameter sets and the reciprocal
$1/\delta$ of the tolerance are polynomial in $1/\varepsilon$ as well as in $m$, while their size,
as in flip-flatness, depends on $r$ alone.

\bibliography{bibliography}
\appendix

\section{Proofs for \Cref{sec:variants}}\label{app:variants}

This appendix formalises the statements summarised in \Cref{tab:variants}. We do not define the
structural notions involved --- nowhere density, treewidth, treedepth, cliquewidth, shrubdepth ---
and refer to \cite{nevsetvril2012sparsity,flipbreak-full} and the references there. The eight
equivalences are those of \cite[Section~18]{flipbreak-full}; what we supply is the reduction of each
packing property to the corresponding variant of flatness or of breakability, and the converse
constructions. Proofs are given in outline where they are modifications of an argument of
\Cref{sec:packing,sec:backtosep}. As there, weightings are normalised to total weight $1$ in proofs.

\subsection{The properties}

Throughout, $\rho$ ranges over $\N\cup\{\infty\}$ and $G$ is a graph. For $F,X\subseteq V(G)$ we use
$\ind^\rho_{G-F}(X)$ as defined in \Cref{sec:variants}, and
\[ \ind^\rho_{\flip,F}(X):=\max_{G'\in\flip(G/F)}\ind^\rho_{G'}(X). \]
Note that $\ind^\rho_{\flip,F}$ is the single-flip quantity of \Cref{def:flipnet}, not the
flip-metric quantity $\ind^\rho_F$; the two satisfy $\ind^\rho_{\flip,F}\le\ind^\rho_F$. We write
$\mathsf{op}$ for either $\del$ or $\flip$, with $\ind^\rho_{\del,F}:=\ind^\rho_{G-F}$, and
$B^\rho_{\mathsf{op},S}(u)$ for the ball of radius $\rho$ around $u$ in $G-S$, respectively in the
flip-metric over $S$; for $u\in S$ and $\mathsf{op}=\del$ we set $B^\rho_{\del,S}(u):=\{u\}$.

\begin{definition}\label{def:app-packing}
Let $\mathsf{op}\in\{\del,\flip\}$ and $\rho\in\N\cup\{\infty\}$. A class $\C$ is
\emph{$\mathsf{op}$-packable at radius $\rho$} if there is $k\in\N$ such that for every
$\varepsilon>0$ and $m\in\N$ there are $\delta>0$ and $t\in\N$ such that for every $G\in\C$ and
every weighting $\w$ of $G$ there is $\Fcal\subseteq[V(G)]^{\le k}$ with $|\Fcal|\le t$ such that
every $\varepsilon$-heavy $X\subseteq V(G)$ of $\delta$-light vertices satisfies
$\ind^\rho_{\mathsf{op},F}(X)\ge m$ for some $F\in\Fcal$. It is
\emph{$2$-$\mathsf{op}$-packable at radius $\rho$} if this holds with $m$ fixed to $2$, so that
$\delta$ and $t$ depend on $\varepsilon$ alone; $k$ still depends on $\rho$ alone.
\end{definition}

Thus \Cref{def:packing} is $\flip$-packability at every finite radius, and \Cref{def:2packing} is
$2$-$\flip$-packability at every finite radius. The two families of properties to which we reduce
these are the following; they are the properties of \cite[Section~18]{flipbreak-full}, transcribed
into our notation.

\begin{definition}\label{def:app-flat-break}
Let $\mathsf{op}$ and $\rho$ be as above.
\begin{enumerate}
\item $\C$ is \emph{$\mathsf{op}$-flat at radius $\rho$} if there are $k\in\N$ and $f:\N\to\N$ such
that for all $G\in\C$, $m\in\N$ and $A\subseteq V(G)$ with $|A|\ge f(m)$ there is $F\subseteq V(G)$
with $|F|\le k$ and $\ind^\rho_{\mathsf{op},F}(A)\ge m$.
\item $\C$ is \emph{$\mathsf{op}$-breakable at radius $\rho$} if there are $k\in\N$ and $f:\N\to\N$
such that for all $G\in\C$, $m\in\N$ and $W\subseteq V(G)$ with $|W|\ge f(m)$ there are
$F\subseteq V(G)$ with $|F|\le k$ and disjoint $A_1,A_2\subseteq W$ of size at least $m$ such that
$A_1$ and $A_2$ lie at distance greater than $\rho$ in $G-F$, respectively in some $F$-flip of $G$.
\end{enumerate}
\end{definition}

At finite radii, $\flip$-flatness and $\flip$-breakability are \Cref{def:flipflat,def:flipbreak}, and
$\del$-flatness is uniform quasi-wideness. One point of bookkeeping recurs. In \cite{flipbreak-full}
a flip is given by a partition of the vertex set into boundedly many parts; for the translation to
definable flips we once again use \Cref{fact:conversion}.

\subsection{Imported results}

\begin{fact}[Flatness]\label{fact:app-flat}
For a class $\C$ of graphs:
\begin{enumerate}
\item $\C$ is $\del$-flat at every finite radius iff $\C$ is nowhere dense \cite{NOdM11};
\item $\C$ is $\flip$-flat at every finite radius iff $\C$ is monadically stable
\cite{indiscernibles};
\item $\C$ is $\del$-flat at radius $\infty$ iff $\C$ has bounded treedepth
\cite[Theorem~18.28]{flipbreak-full};
\item $\C$ is $\flip$-flat at radius $\infty$ iff $\C$ has bounded shrubdepth
\cite[Theorem~18.17]{flipbreak-full}.
\end{enumerate}
\end{fact}

\begin{fact}[Breakability]\label{fact:app-break}
For a class $\C$ of graphs:
\begin{enumerate}
\item $\C$ is $\del$-breakable at every finite radius iff $\C$ is nowhere dense
\cite[Theorem~18.2]{flipbreak-full};
\item $\C$ is $\flip$-breakable at every finite radius iff $\C$ is monadically dependent
\cite[Theorem~13.1]{flipbreak-full};
\item $\C$ is $\del$-breakable at radius $\infty$ iff $\C$ has bounded treewidth
\cite[Theorem~18.12]{flipbreak-full};
\item $\C$ is $\flip$-breakable at radius $\infty$ iff $\C$ has bounded cliquewidth
\cite[Theorem~18.4]{flipbreak-full}.
\end{enumerate}
\end{fact}

\begin{fact}[Depth parameters]\label{fact:app-depth}
Let $\rk^\del_\rho$ be defined as in \Cref{def:rank}, with $B^\rho_{S\cup\{s\}}(u)$ replaced by
$B^\rho_{\del,S\cup\{s\}}(u)$.
\begin{enumerate}
\item $\rk^\del_r$ is bounded on $\C$ for every finite $r$ iff $\C$ is nowhere dense
\cite{deciding};\footnote{The Splitter game of \cite{deciding} is stated in its shrinking form, in
which the arena becomes the ball around Localiser's vertex in the current arena rather than in
$G-S$. That the confining form used here is likewise bounded exactly on nowhere dense classes follows
as for flips, cf.\ \cite[Section~3]{flippergame-full} and \Cref{app:algorithmic}.

}
\item $\rk_r$ is bounded on $\C$ for every finite $r$ iff $\C$ is monadically stable
\cite{flippergame,flippergame-full};
\item $\td(G)\le\rk^\del_\infty(G)\le\td(G)+1$, where $\td(G)$ is the treedepth of $G$
\cite{nevsetvril2012sparsity};
\item every class of graphs of bounded shrubdepth has bounded $\rk_\infty$.\footnote{The converse
holds too, e.g.\ by the fact that bounded $\rk_\infty$ gives $\flip$-packability at radius $\infty$
by \Cref{lem:app-rank-packing}, hence $\flip$-flatness at radius $\infty$ by
\Cref{lem:app-packing-flat}, hence bounded shrubdepth by \Cref{fact:app-flat}(4).}
\end{enumerate}
\end{fact}

The last item is folklore. One way to see it is to recall that bounded shrubdepth is equivalent to
bounded \emph{flipdepth} \cite[Fact~18.18]{flipbreak-full}, where $K_1$ has flipdepth $0$ and $G$ has
flipdepth at most $h+1$ if it is a $2$-flip of a disjoint union of graphs of flipdepth at most $h$.

Given such a $G$, apply \Cref{fact:conversion}(1) to the $2$-flip: it yields a set $S_0$, of size
bounded in terms of the VC-dimension, whose definable flip has components refining those of the
disjoint union. Flipper plays the elements of $S_0$ in turn, after which the arena lies inside a
single component, which is a graph of flipdepth at most $h$ up to a flip over the parameters played;
the argument repeats, with the partition refined at each level. Hence $\rk_\infty(G)$ is bounded in
terms of $h$ and the VC-dimension, both bounded on the class.

\begin{fact}[Balanced separators]\label{fact:app-sep}
\begin{enumerate}
\item If $G$ has a tree decomposition of width $w$ and $\w$ is a weighting of $G$, then some bag $B$
satisfies $|B|\le w+1$ and $\w(C)\le\frac12\w(V(G))$ for every connected component $C$ of $G-B$;
this follows from \cite{reed_survey}. Orienting each edge of the decomposition tree towards the
heavier side and taking a sink bag proves it.
\item If $G$ has rankwidth at most $w$ and $\w$ is a $\frac14$-balanced weighting of $G$, then there
are a partition $\Pcal$ of $V(G)$ with $|\Pcal|\le2^w+2^{2^w}$ and a $\Pcal$-flip $H$ of $G$ with
$\w(C)\le\frac34\w(V(G))$ for every connected component $C$ of $H$. This is
\cite[Lemmas~18.6 and 18.7]{flipbreak-full} with the counting measure replaced by $\w$: the descent
of \cite[Lemma~18.6]{flipbreak-full} produces an edge of the decomposition tree whose two sides
$X,Y$ each carry at least a quarter of the weight, and \cite[Lemma~18.7]{flipbreak-full} produces a
$\Pcal$-flip with no edges between $X$ and $Y$, so every component lies inside $X$ or inside $Y$.
\end{enumerate}
\end{fact}

\begin{fact}[Bipartite Ramsey]\label{fact:bipramsey}
There is a function $\Rcal_2:\N\times\N\to\N$ such that for all $m,c\in\N$, every colouring with $c$
colours of the edges of the complete bipartite graph with both sides of size $\Rcal_2(m,c)$ admits
sides $A'$, $B'$ of size $m$ all of whose edges receive the same colour.
\end{fact}

\subsection{From packability to flatness and breakability}

\begin{lemma}\label{lem:app-packing-flat}
If $\C$ is $\mathsf{op}$-packable at radius $\rho$, then $\C$ is $\mathsf{op}$-flat at radius
$\rho$, with the same $k$.
\end{lemma}

\begin{proof}
As in \Cref{lem:packing-flat}: given $m$, put $\delta:=\delta(1,m)$ and
$f(m):=\lceil1/\delta\rceil+1$. For $A\subseteq V(G)$ with $|A|\ge f(m)$, every vertex is
$\delta$-light for the indicator weighting $\w_A$ and $A$ itself is $1$-heavy, so some $F\in\Fcal$
satisfies $\ind^\rho_{\mathsf{op},F}(A)\ge m$.
\end{proof}

\begin{lemma}\label{lem:app-m2-break}
If $\C$ is $2$-$\mathsf{op}$-packable at radius $4\rho$, read as $\infty$ when $\rho=\infty$, then
$\C$ is $\mathsf{op}$-breakable at radius $\rho$.
\end{lemma}

\begin{proof}
The argument of \Cref{lem:light-balls} applies verbatim to both operations and at both radii: from
$2$-packability at radius $4\rho$ with tolerance $\frac12$ it yields $k'$ and $\delta$ such that for
every $G\in\C$ and every weighting $\w$ there is $S$ with $|S|\le k'$ such that
$B^{2\rho}_{\mathsf{op},S}(u)$ is not $\frac12$-heavy for every $\delta$-light $u$. For deletions
the contradiction reads: two vertices $a,b$ of the ball, hence not in $S$, lie at distance greater
than $4\rho$ in $G-F$ for some $F\subseteq S$, hence in $G-S$, contradicting the triangle
inequality; at $\rho=\infty$ it reads: two elements of one component over $S$ lie in distinct
components over some $F\subseteq S$, which is impossible as components over $S$ refine those over
$F$.

The rest is a variant of \cite[Lemma~24]{separability}. Let $W\subseteq V(G)$ with
$|W|\ge f(m):=\max(4m^2,\lceil1/\delta\rceil+1)$, and apply the above to the indicator weighting of
$W$, for which every vertex is $\delta$-light. For finite $\rho$: if some $u$ has
$|B^\rho_{\mathsf{op},S}(u)\cap W|\ge2m$, take $A_1:=B^\rho_{\mathsf{op},S}(u)\cap W$ and
$A_2:=W\setminus B^{2\rho}_{\mathsf{op},S}(u)$, which has more than $\frac12|W|\ge m$ elements; the
triangle inequality puts every pair from $A_1\times A_2$ at distance greater than $\rho$. Otherwise
every $\rho$-ball meets $W$ in fewer than $2m$ elements, so greedily choosing points of $W$ at
pairwise distance greater than $\rho$ yields at least $|W|/2m\ge2m$ of them, which we split in half.
For $\rho=\infty$, every component meets $W$ in fewer than half of its elements, so ordering the
components arbitrarily and letting $A_1$ be $W$ intersected with the shortest initial segment
meeting $W$ in at least $\frac14|W|$ elements leaves $A_2:=W\setminus A_1$ with more than
$\frac14|W|\ge m$ elements, in a disjoint union of other components.

For $\mathsf{op}=\del$ the two parts are far apart in $G-S$ and we are done. For
$\mathsf{op}=\flip$ they are far apart in the flip-metric over $S$, so each pair is separated by
some member of $\flip(G/S)$, of which there are at most $\Xi(k')$. Colouring the pairs of
$A_1\times A_2$ accordingly and applying \Cref{fact:bipramsey}, with $f(m)$ enlarged to
$\max(4\Rcal_2(m,\Xi(k'))^2,\lceil1/\delta\rceil+1)$, gives subsets of size $m$ all of whose pairs
are separated by one and the same flip.
\end{proof}

\subsection{From the structural notions to packability}

\begin{lemma}\label{lem:app-rank-packing}
Let $\rho\in\N\cup\{\infty\}$ and read $3\rho$ as $\infty$ when $\rho=\infty$. If $\rk^\del_{3\rho}$
is bounded on $\C$, then $\C$ is $\del$-packable at radius $\rho$; if $\rk_{3\rho}$ is bounded on
$\C$, then $\C$ is $\flip$-packable at radius $\rho$.
\end{lemma}

\begin{proof}
The proof of \Cref{lem:winning-packing} goes through with the balls of the flip-metric replaced by
the corresponding balls. Its first stage uses only three properties of these: that playing the
vertex realising the minimum in the rank recursion drops the rank of every ball of radius $3\rho$
centred in the arena; that at most $m/\varepsilon$ pairwise disjoint balls of weight at least
$\varepsilon/m$ fit; and that a ball of radius $\rho$ meeting another lies inside the ball of radius
$3\rho$ around its centre. The first is the definition of the rank, the second is insensitive to the
operation and the radius, and the third is the triangle inequality; at $\rho=\infty$ it is the
statement that two intersecting components coincide. For $\mathsf{op}=\flip$ the second stage
applies verbatim, with $\Xi(k)$ colours in the Ramsey argument.

For $\mathsf{op}=\del$ the second stage is not needed, but two adjustments are, because a deleted
vertex is not counted whereas an isolated one is. First, the base case of the induction, in which
the arena lies inside the parameters played, must be vacuous: with $K$ bounding the rank, at most $K$
parameters are played along any branch, so taking $\delta_0(\varepsilon,m):=\varepsilon/(Km)$ in
place of $\varepsilon/m$, and hence $\delta_k(\varepsilon,m)=\varepsilon/(Km^{k+1})$, makes any set
of $\delta_k$-light vertices contained in the parameters too light to be $\varepsilon$-heavy.
Second, in the case where the parameters $S'$ played so far already pull $m$ elements of
$X\setminus S'$ apart in $G-S'$, the member of $\Fcal$ to use is $\{s\}$ itself, so that no further
deletion can remove those elements; and in the other case the fewer than $m$ balls of the
$S'$-metric covering $X\setminus S'$ carry weight at least $\varepsilon-K\delta\ge\varepsilon(1-\frac1m)$,
as $X\cap S'$ consists of at most $K$ light vertices, so one of them still carries at least
$\frac\varepsilon m$. With these changes the induction delivers $\ind^\rho_{G-F}(X)\ge m$ directly.
\end{proof}

\begin{proposition}\label{prop:app-allm}
For a class $\C$ of graphs:
\begin{enumerate}
\item $\C$ is $\del$-packable at every finite radius iff $\C$ is nowhere dense;
\item $\C$ is $\flip$-packable at every finite radius iff $\C$ is monadically stable;
\item $\C$ is $\del$-packable at radius $\infty$ iff $\C$ has bounded treedepth;
\item\label{it:4} $\C$ is $\flip$-packable at radius $\infty$ iff $\C$ has bounded shrubdepth.
\end{enumerate}
\end{proposition}

\begin{proof}
Item (2) is \Cref{thm:main}. In each of the others, the forward direction is
\Cref{lem:app-packing-flat} followed by the corresponding item of \Cref{fact:app-flat}. The backward
direction is \Cref{lem:app-rank-packing}, fed by \Cref{fact:app-depth}(1) in case (1), by
\Cref{fact:app-depth}(3) in case (3), and by \Cref{fact:app-depth}(4) in case (4).
\end{proof}

\subsection{The $m=2$ case}

\begin{lemma}\label{lem:app-halving}
\begin{enumerate}
\item If $\C$ has treewidth at most $w$ then $\C$ is $2$-$\del$-packable at radius $\infty$, with
$k=w+1$.
\item If $\C$ has rankwidth at most $w$ then $\C$ is $2$-$\flip$-packable at radius $\infty$, with
$k$ bounded in terms of $w$ and the VC-dimension.
\end{enumerate}
\end{lemma}

\begin{proof}
Throughout, write $G/F$ for $G-F$ in case (1), and for the flip of $G$ over $F$ supplied below in
case (2); in both cases the components of $G/F$ are what the recursion descends into, and $F$ is a
set of at most $k$ vertices. Fix $\varepsilon>0$, put $L:=\lceil\log_{4/3}(4/\varepsilon)\rceil$ and
$\delta:=\varepsilon/(4kL)$, with $k$ as determined below; note $\delta\le\varepsilon/16$. Let
$\w$ be a weighting of $G$, and let $\w'$ be its restriction to the $\delta$-light vertices, that
is, $\w'$ agrees with $\w$ on them and is $0$ on the $\delta$-heavy ones. Every set to be served
consists of $\delta$-light vertices, so its weight is the same under $\w$ and $\w'$; and if
$\w'(V(G))<\varepsilon$ there is nothing to serve. So assume $\w'(V(G))\ge\varepsilon$, and work with
$\w'$, whose atoms are all less than $\delta$.

Call $F$ a \emph{halving step} for $D\subseteq V(G)$ with $\w'(D)>0$ if
\[ \w'(C\cap D)\le\tfrac34\,\w'(D)\qquad\text{for every component }C\text{ of }G/F. \]
Such a step exists in both cases, by applying \Cref{fact:app-sep} to the restriction of $\w'$ to
$D$. In case (1), fix once and for all a tree decomposition of $G$ of width $w$; some bag $F$, of
size at most $w+1$, leaves every component of $G-F$ with at most half of $\w'(D)$. Note that the
components in question are those of the whole graph minus $F$, not of the induced subgraph on $D$;
this is what allows a single fixed decomposition to serve at every level of the recursion. In
case (2), \Cref{fact:app-sep}(2) requires the restriction of $\w'$ to $D$ to be $\frac14$-balanced,
which holds for the sets $D$ below, as their weight is at least $\varepsilon/4$ while every atom is
less than $\delta\le\varepsilon/16$; it gives a partition $\Pcal$ with $|\Pcal|\le2^w+2^{2^w}$ and a
$\Pcal$-flip in which every component meets $D$ in at most $\frac34\w'(D)$, and
\Cref{fact:conversion}(1) gives a set $F$ of size $\Ocal(d|\Pcal|^2)$ and an $F$-flip whose
components refine those of the $\Pcal$-flip, so the same bound holds for it. In both cases
$|F|\le k$ for a $k$ depending on $w$ and, in case (2), on the VC-dimension $d$, but not on
$\varepsilon$.

Build a rooted tree of sets: the root is $V(G)$; a node $D$ at depth less than $L$ with
$\w'(D)\ge\varepsilon/4$ is assigned a halving step $F_D$ as above, and its children are the sets
$D\cap C$ of weight at least $\varepsilon/4$, for $C$ a component of $G/F_D$. Nodes of weight below
$\varepsilon/4$, and all nodes at depth $L$, are leaves. The children of a node are disjoint, so
there are at most $4/\varepsilon$ of them, and the tree has at most $(4/\varepsilon)^{L+1}$ nodes.
Let $\Fcal:=\{F_D:D\text{ an internal node}\}$, a family of at most
$t(\varepsilon):=(4/\varepsilon)^{L+1}$ sets, each of size at most $k$.

We claim $\Fcal$ witnesses $2$-packability at radius $\infty$. Let $X$ be $\varepsilon$-heavy with
$\delta$-light vertices, so that $\w'(X)=\w(X)\ge\varepsilon$, and suppose for contradiction that
$\ind^\infty_{\mathsf{op},F}(X)\le1$ for every $F\in\Fcal$; that is, $X\setminus F$ in case (1),
respectively $X$ in case (2), lies within a single component of $G/F$. Set $D_0:=V(G)$ and
$X_0:=X$, and suppose $X_i\subseteq D_i$ with $D_i$ an internal node. Put
$X_{i+1}:=X_i\setminus F_{D_i}$; since $|F_{D_i}|\le k$ and the vertices of $X$ are $\delta$-light,
$\w'(X_{i+1})\ge\w'(X_i)-k\delta$. By assumption $X_{i+1}$ lies in one component $C$ of
$G/F_{D_i}$, so $X_{i+1}\subseteq D_{i+1}:=D_i\cap C$, and $\w'(D_{i+1})\le\frac34\w'(D_i)$ by the
choice of $F_{D_i}$.

Along this path, $\w'(X_i)\ge\varepsilon-ik\delta\ge\frac34\varepsilon$ for every $i\le L$, by the
choice of $\delta$. Hence $\w'(D_i)\ge\w'(X_i)\ge\frac34\varepsilon>\varepsilon/4$, so each $D_i$
with $i<L$ is an internal node and the descent continues. But
$\w'(D_L)\le(\frac34)^L\le\varepsilon/4$ by the choice of $L$, contradicting
$\w'(D_L)\ge\frac34\varepsilon$. So some $F\in\Fcal$ satisfies $\ind^\infty_{\mathsf{op},F}(X)\ge2$.
\end{proof}

\begin{proposition}\label{prop:app-m2}
For a class $\C$ of graphs:
\begin{enumerate}
\item $\C$ is $2$-$\del$-packable at every finite radius iff $\C$ is nowhere dense;
\item $\C$ is $2$-$\flip$-packable at every finite radius iff $\C$ is monadically dependent;
\item $\C$ is $2$-$\del$-packable at radius $\infty$ iff $\C$ has bounded treewidth;
\item $\C$ is $2$-$\flip$-packable at radius $\infty$ iff $\C$ has bounded cliquewidth.
\end{enumerate}
\end{proposition}

\begin{proof}
Item (2) is \Cref{thm:2packing-equiv}. In each of the others the forward direction is
\Cref{lem:app-m2-break} followed by the corresponding item of \Cref{fact:app-break}. For the backward
directions: (1) follows from \Cref{prop:app-allm}(1), $2$-packability being an instance of the full
property; (3) and (4) are the two items of \Cref{lem:app-halving}, using for (4) that bounded
cliquewidth and bounded rankwidth coincide \cite[Fact~18.5]{flipbreak-full} and that classes of
bounded cliquewidth have bounded VC-dimension.
\end{proof}

\Cref{prop:app-allm,prop:app-m2} together prove \Cref{thm:everything}.

\section{Proofs for \Cref{sec:algorithmic}}\label{app:algorithmic}

\subsection{The algorithmic Flipper game}

Recall from \Cref{def:rank} that $\rk_r(G)$ is the value of a game in which Flipper plays single
vertices, accumulating a parameter set $S$, and Localiser answers with a vertex $c$ of the current
arena, which then shrinks to its intersection with $B^r_S(c)$; Flipper wins once nothing but
parameters remains. This is the \emph{confining Flipper game with quantifier-free definable
separation} of \cite[Section~3.1]{flippergame-full}: there the arena after round $k$ is
$A_k=A_{k-1}-\{w:w\forkindep^r_{S_{k-1}}c_k\}$, and their relation $w\forkindep^r_Sc$ holds
precisely when some $S$-flip puts $w$ and $c$ at distance greater than $r$, that is, precisely when
$\dist_S(w,c)>r$ in our notation. So their arena is $A_{k-1}\cap B^r_{S_{k-1}}(c_k)$ and $\rk_r(G)$
is the number of rounds Flipper needs in their game.

The algorithmic result of \cite{flippergame-full} concerns the original form of the game, the
\emph{shrinking Flipper game with atomic flips}: there Localiser restricts the current arena to a
ball of radius $r$ in it, and Flipper responds with an atomic flip, that is, with a pair of vertex
sets whose adjacency is complemented. We write $G\oplus F$ for the graph obtained from $G$ by
performing the atomic flip $F$, and $G\oplus F_1\oplus\dots\oplus F_i$ for the graph after a
sequence of them. In \cite[Remark~11.1]{flippergame-full} the game is relaxed to the
\emph{Induced-Subgraph-Flipper game}, in which Localiser may localise to an arbitrary induced
subgraph of such a ball rather than to the whole ball, and it is for that variant that the
algorithmic result is proved. The notation below is exactly that of \cite{flippergame-full}.

\begin{fact}[{\cite[Theorem~11.2]{flippergame-full}}]\label{fact:algo-flipper}
Let $\C$ be a monadically stable class of graphs and $r\in\N$. There are $\ell\in\N$ and a
radius-$r$ Flipper strategy $\mathrm{flip}^\star$ for the Induced-Subgraph-Flipper game that is
$\ell$-winning on $\C$ and whose moves on an $n$-vertex graph $G\in\C$ can be computed in time
$\Ocal_{\C,r}(n^2)$, where $n=|V(G)|$ refers to the original graph rather than to the current
arena.
\end{fact}

We shall simulate $\mathrm{flip}^\star$ in the confining game, using \Cref{fact:conversion}(1) to
turn each atomic flip it proposes into a bounded set of vertices, to be played by Flipper in
consecutive rounds. The conversion is a statement about the whole graph: the definable flip it
returns approximates the current flipped graph, so the arena it leaves us with is contained in a
ball of that graph, and not in a ball of the current arena, as Localiser's moves in
\Cref{fact:algo-flipper} are required to be. We therefore need the following lemma, which follows by
analysing the proof of \Cref{fact:algo-flipper}.

\begin{lemma}\label{lem:algo-simulate}
Let $\C$, $r$, $\mathrm{flip}^\star$ and $\ell$ be as in \Cref{fact:algo-flipper}. Then
$\mathrm{flip}^\star$ remains $\ell$-winning against any Localiser whose response, after Flipper's
flips $F_1,\dots,F_i$, is a set contained in some ball of radius $r$ of
$G\oplus F_1\oplus\dots\oplus F_i$.
\end{lemma}

\begin{proof}
The strategy $\mathrm{flip}^\star$ is analysed through \cite[Claim~11.3]{flippergame-full}, whose
argument uses only the following property of Localiser's move: in the era in which Flipper plays the
move pair defined by $F_Z:=\mathrm{Predict}_{2r}(G,\preccurlyeq,Z)$, the resulting arena is
contained in a ball of radius $r$ of the current flipped graph, whence $y_6$ and $y_7$ cannot both
survive, as $Y$ is distance-$2r$ independent there. There that containment is deduced from the move
pairs keeping the arena an induced subgraph of $G$; here it is supplied by \Cref{fact:conversion}(1).
The bound $\ell$ and the running time are unaffected, since Flipper's moves are computed from $G$,
the order $\preccurlyeq$ and the set $X$ of vertices played, none of which depends on how Localiser
chooses within the ball.
\end{proof}

\begin{proof}[Proof of \Cref{thm:rank-algo}]
Let $d$ bound the VC-dimension of the graphs in $\C$, which is finite as $\C$ is monadically
stable. Invoke \Cref{fact:algo-flipper} at radius $R:=5r$, obtaining $\ell$ and
$\mathrm{flip}^\star$, and simulate $\mathrm{flip}^\star$ in the confining game.

Suppose $\mathrm{flip}^\star$ has proposed atomic flips $F_1,\dots,F_i$, and let
$G_i:=G\oplus F_1\oplus\dots\oplus F_i$. Each atomic flip is determined by a pair of vertex sets, so
$G_i$ is a flip of $G$ over a partition $\Pcal_i$ of $V(G)$ into at most $4^i$ parts, and
\Cref{fact:conversion}(1) provides $T_i\subseteq V(G)$ with $|T_i|\le\Ocal(d\cdot16^i)$ and
$H_i\in\flip(G/T_i)$ with $B^\rho_{H_i}(v)\subseteq B^{5\rho}_{G_i}(v)$ for all $v$ and $\rho$.
Flipper plays the elements of $T_i$ one at a time. After these rounds the arena is contained in
$B^r_S(c)$ for Localiser's last move $c$ and the current parameter set $S\supseteq T_i$, and
\[ B^r_S(c)\;\subseteq\;B^r_{T_i}(c)\;\subseteq\;B^r_{H_i}(c)\;\subseteq\;B^{5r}_{G_i}(c), \]
the second inclusion because $H_i\in\flip(G/T_i)$ and the flip-metric is the maximum over
$\flip(G/T_i)$. So the arena is contained in a ball of radius $R$ of $G_i$, and
\Cref{lem:algo-simulate} applies: the simulated play is one Flipper wins within $\ell$ moves of
$\mathrm{flip}^\star$, hence within $k:=\sum_{i\le\ell}|T_i|=\Ocal(d\cdot16^\ell)$ rounds of the
confining game. Each move of $\mathrm{flip}^\star$ costs time $\Ocal_{\C,R}(n^2)$, followed by the
conversion of \Cref{fact:conversion}(1), which inspects each pair of vertices once.

\end{proof}

\subsection{Computing the flip-net}

\begin{proof}[Proof of \Cref{thm:algo}]
Let $\w$ be the given weighting, and follow the proof of \Cref{lem:winning-packing} with
$m':=\Rcal(m,\Xi(k))$ in place of $m$, as in its Stage~2; weights are compared against fractions of
$\w(V(G))$, which is computed once. The data at a node of the recursion is a pair $(A',S')$ with
$|S'|\le k$, where $k$ is the bound of \Cref{thm:rank-algo} at radius $3r$, together with the
current tolerance $\varepsilon'=\varepsilon/m'^{\,j}$ at depth $j$; we account for the cost at one
node.

\emph{Flipper's move.} The vertex realising the minimum in the rank recursion is supplied by
\Cref{thm:rank-algo} at radius $3r$ in time $\Ocal_{\C,r}(n^2)$.

\emph{The flip-metric.} Enumerate $\flip(G/S')$, of size at most $\Xi(k)$, each member being
determined by a set of pairs of $S'$-classes and computable in time $\Ocal(n^2)$. Compute the balls
of radius $3r$ around every vertex in each by breadth-first search, in time $\Ocal(n^3)$ per flip;
intersecting over the family gives $B^{3r}_{S'}(v)$ for every $v$, and $B^r_{S'}(v)$ likewise. This
costs $\Ocal(\Xi(k)\cdot n^3)$ and dominates.

\emph{The packing.} Compute $\w(B^r_{S'}(v))$ for every $v\in A'$, each a sum over at most $n$
vertices, in $\Ocal(n^2)$ arithmetic operations. Then scan the vertices of $A'$, maintaining an
inclusion-maximal family of pairwise disjoint balls $B^r_{S'}(v)$, $v\in A'$, of weight at least
$\frac{\varepsilon'}{m'}\w(V(G))$; each test is a comparison and a disjointness check against the
balls already chosen. At most $\lfloor m'/\varepsilon'\rfloor$ balls are chosen, so this costs
$\Ocal(\frac{m'}{\varepsilon'}n^2)$.

The children of the node are the pairs $(A'\cap B^{3r}_{S'}(v_i),S'\cup\{s\})$ for the centres
$v_i$ of the packed balls, so the branching at depth $j$ is at most $m'^{\,j+1}/\varepsilon$; the
rank drops at each step, so the depth is at most $k$. The recursion tree therefore has at most
$\prod_{j<k}m'^{\,j+1}/\varepsilon\le(m'^k/\varepsilon)^k$ nodes, and the total cost is
\[ \Ocal\big((m'^k/\varepsilon)^k\cdot\Xi(k)\cdot n^3\big)=\Ocal_{\C,r,\varepsilon,m}(n^3). \]
The parameter sets collected at the nodes form $\Fcal$, of size at most $t$, each of size at most
$k$; by Stage~1 of \Cref{lem:winning-packing}, every $\varepsilon$-heavy set of $\delta$-light
vertices has $m'$ elements pairwise far apart in the flip-metric over some $F\in\Fcal$. Finally,
Stage~2 applies Ramsey's theorem to these $m'$ elements, a quantity independent of $n$, with the
colouring read off from the flips already computed; this is performed when a heavy set is
presented and does not enter the construction of $\Fcal$.
\end{proof}

\subsection{Computing flip-nets on monadically dependent classes}\label{app:dependent}

Here we prove \Cref{thm:algo-dependent,cor:algo-sep}. The family of
\Cref{lem:dependent-packing} is obtained from the compression argument of
\cite[Lemma~20]{separability}, which starts from a family of size linear in the graph and shrinks
it. At first sight that argument is not effective: the sets it discards at each step are chosen
through a colouring of tuples supplied by Gaifman's locality theorem, and computing that colouring
amounts to computing local first-order types, which is essentially the model-checking problem on
$\C$. We show that the colouring is needed only in the analysis. Choosing the sets to keep greedily,
and using locality to bound the number of greedy steps rather than to make them, turns the
compression into an algorithm whose only external ingredient is flip-breakability in the algorithmic
form of Dreier, M\"ahlmann and Toru\'nczyk.

\begin{fact}[{\cite[Theorems~13.2 and~19.14]{flipbreak-full}} with {\cite[Lemma~4]{separability}}]
\label{fact:algo-break}

Let $\C$ be a monadically dependent class of graphs and $r\in\N$. There are $t\in\N$, an unbounded
function $M_r:\N\to\N$, and an algorithm that, given $G\in\C$ and sets $W_1,W_2\subseteq V(G)$ with
$|W_1|,|W_2|\ge M_r(m)$, computes in time $\Ocal_{\C,r}(n^2)$ sets $A_1\subseteq W_1$,
$A_2\subseteq W_2$ with $|A_1|,|A_2|\ge m$ and $S\subseteq V(G)$ with $|S|\le t$ such that
$\dist_S(A_1,A_2)>2r$.
\end{fact}

The flip returned by \cite[Theorem~13.2]{flipbreak-full} is a partition flip; it is converted into a
definable one by \Cref{fact:conversion}(1), whose proof is constructive
\cite[Section~3]{separability}, at a constant factor in the radius which we absorb into $r$.

\begin{proof}[Proof of \Cref{thm:algo-dependent}]
By the proof of \Cref{lem:dependent-packing} it suffices to compute, for the weighting $\w'$
obtained from $\w$ by giving weight $0$ to the $\delta$-heavy vertices with
$\delta=\varepsilon^2/2$ --- one pass over $V(G)$ --- a family $\Fcal$ as in
\Cref{fact:sparsifying-family} at radius $r$ and tolerance $\varepsilon/2$: at most $t$ sets of size
at most $k$ with $\w'(B^r_{\Fcal}(v))\le\frac\varepsilon2\w'(V(G))$ for every $v$, where
$B^r_\Fcal(v)=\bigcap_{S\in\Fcal}B^r_S(v)$. We follow the proof of \cite[Lemma~20]{separability},
which produces such an $\Fcal$ by repeatedly shrinking a family with that property, and check that
each step can be carried out in polynomial time. We keep the notation of that proof, writing
$\Fcal$ for the current family, all of whose members have size $k$, and calling $\Fcal$
\emph{sparsifying} if it satisfies the displayed condition.

\emph{Initialisation.} Put $\Fcal:=\{S_v:v\in V(G)\}$ with $S_v\ni v$ an arbitrary $k$-set. This is
sparsifying, as $v$ is isolated in a flip over $S_v$ and $\w'$ is $\frac\varepsilon2$-balanced.

\emph{One compression step.} Let $\Fcal$ be sparsifying with $|\Fcal|\ge k_0$, where $k_0$ is the
threshold of \cite[Lemma~20]{separability}. The proof first applies \cite[Lemma~19]{separability} to
$\Fcal$, obtaining $\Fcal'\subseteq\Fcal$ of unbounded size and a family $\Ycal$ of at most
$\lambda$ sets of size $k$ with
\[ \w'\big(B^{3r'}_{\Ycal}(S\setminus\textstyle\bigcup\Ycal)\big)\le\tfrac\varepsilon2\,\w'(V(G))
   \qquad\text{for every }S\in\Fcal', \tag{$*$} \]
where $r'$ is the locality radius of \cite[Lemma~14]{separability} for formulas of quantifier rank
$r$. The proof of \cite[Lemma~19]{separability} consists of the sunflower lemma, whose proof by
Erd\H{o}s and Rado is a greedy algorithm, followed by $\binom p2k^2$ applications of
flip-breakability to pairs of sets of coordinates, each supplied by \Cref{fact:algo-break} in time
$\Ocal(n^2)$; the balls $B^{3r'}_{\Ycal}$ are computed by enumerating the at most $\Xi(k)$ flips
over each member of $\Ycal$ and running breadth-first search. So $\Fcal'$ and $\Ycal$ are computed
in polynomial time.

The proof then removes from $\Fcal'$ all but a bounded subfamily $\overline{\Xcal}$, chosen by
keeping one member of $\Fcal'$ per value of a colouring $\mathrm{col}_2$ of $k$-tuples supplied by
\cite[Lemma~14]{separability}, and sets
$\Fcal^*:=(\Fcal\setminus\Fcal')\cup\overline{\Xcal}\cup\Ycal$. We choose $\overline{\Xcal}$
differently. Call a vertex $v$ \emph{near} if $v\in B^{2r'}_{\Ycal}(\bigcup\Fcal')$ and \emph{far}
otherwise, which is decided from the balls already computed. Starting from
$\overline{\Xcal}:=\emptyset$, repeat:
\begin{quote}
if there are a far vertex $v$ and a vertex $w$ with $\dist_{\Fcal}(v,w)>r$ but
$\dist_{\Fcal^*}(v,w)\le r$, pick $S\in\Fcal'$ with $\dist_S(v,w)>r$ and add it to
$\overline{\Xcal}$;
\end{quote}
stopping when no such pair exists. Such an $S$ exists: $\dist_\Fcal(v,w)>r$ is witnessed by some
member of $\Fcal$, while every member of $\Fcal$ outside $\Fcal'$, and every member of
$\overline{\Xcal}$, already lies in $\Fcal^*$; so the witness lies in
$\Fcal'\setminus\overline{\Xcal}$. Each round costs polynomial time, the distances $\dist_S(v,w)$
for $S\in\Fcal\cup\Ycal$ and all pairs $v,w$ being computed once by breadth-first search in the at
most $\Xi(k)$ flips over each $S$.

On termination $\Fcal^*$ is sparsifying: for a near $v$ this is \cite[Claim~21]{separability}, which
uses only $(*)$ and $\Ycal\subseteq\Fcal^*$, and for a far $v$ the stopping condition gives
$B^r_{\Fcal^*}(v)\subseteq B^r_{\Fcal}(v)$, while $\Fcal$ is sparsifying.

It remains to bound the number of rounds, and this is where locality enters. Let $\Pcal$ be the
common refinement of the partitions $\Pcal_S$, $S\in\Ycal$, which has at most $(k+2^k)^\lambda$
parts, and let $\mathrm{col}_1,\mathrm{col}_2$ be the colourings of pairs and of $k$-tuples supplied
by \cite[Lemma~14]{separability} for $\Pcal$ and the formula $\varphi(vw,\bar s)$ expressing
$\dist_{\{\bar s\}}(v,w)>r$, with $\ell$ colours. Suppose $S$ is added in some round, for the pair
$v,w$. As $\Ycal\subseteq\Fcal^*$ we have $\dist_\Ycal(v,w)\le r$, and as $v$ is far, the triangle
inequality gives $\dist_\Ycal(vw,\bar s')>r'$ for every enumeration $\bar s'$ of every
$S'\in\Fcal'$. Since $\dist_\Pcal\ge\dist_\Ycal$, \cite[Lemma~14]{separability} then says that
whether $\dist_{S'}(v,w)>r$ holds for $S'\in\Fcal'$ depends only on $\mathrm{col}_2(\bar s')$. Every
member of $\overline{\Xcal}$ fails to separate $v$ and $w$, while $S$ separates them, so $S$ carries
a colour that no member of $\overline{\Xcal}$ carries. Hence there are at most $\ell$ rounds, and
\[ |\Fcal^*|\;\le\;|\Fcal|-|\Fcal'|+\ell+\lambda\;<\;|\Fcal|, \]
exactly as in \cite[Lemma~20]{separability}, whose threshold $k_0$ makes the last inequality hold.
Note that $\mathrm{col}_1$ and $\mathrm{col}_2$ are never computed.

\emph{The loop.} Each compression step removes at least one member from a family of size at most
$n$, so there are at most $n$ steps, after which $|\Fcal|<k_0=:t$. Every step runs in polynomial
time with constants depending on $\C$, $r$ and $\varepsilon$, and so does the whole computation. A
naive count gives $\Ocal_{\C,r,\varepsilon}(n^5)$: a linear number of steps, each dominated by the
all-pairs distance tables in the flips over the members of $\Fcal$.
\end{proof}

\begin{proof}[Proof of \Cref{cor:algo-sep}]
Let $\delta$ and $t$ be the constants provided at radius $12r$, tolerance $\varepsilon$ and $m=2$
by \Cref{thm:main} if $\C$ is monadically stable, and by \Cref{lem:dependent-packing} otherwise, and
let $k_0$ bound the size of a member of the corresponding net. We follow the proofs of
\Cref{lem:light-balls,lem:packing-sep}, checking that each step is effective.

Compute a $(12r,\varepsilon,\delta,2)$-flip-net $\Fcal$ of order $(k_0,t)$ for $\w$: by
\Cref{thm:algo} in time $\Ocal_{\C,r,\varepsilon}(n^3)$ if $\C$ is monadically stable, and by
\Cref{thm:algo-dependent} in time $\Ocal_{\C,r,\varepsilon}(n^5)$ in general. One pass over $V(G)$
gives the set $H$ of $\delta$-heavy vertices, of size at most $1/\delta$; put
$S:=\bigcup\Fcal\cup H$, so that $|S|\le k_0t+\lfloor1/\delta\rfloor$. By the proof of
\Cref{lem:light-balls}, $B^{6r}_S(u)$ is not $\varepsilon$-heavy for every $\delta$-light $u$; and
an $\varepsilon$-light $u$ that is $\delta$-heavy lies in $S$, so $B^{6r}_S(u)=\{u\}$ by the
isolating flip, which is not $\varepsilon$-heavy either. Thus $B^{6r}_S(u)$ is not
$\varepsilon$-heavy for every $\varepsilon$-light $u$.

It remains to convert the flip-metric over $S$ into a single flip. \Cref{fact:conversion}(2),
applied to the partition $\Pcal_S$ into $S$-classes, supplies $T\supseteq S$ with $|T|\le k$ and
$G'\in\flip(G/T)$ with $B^r_{G'}(u)\subseteq B^{6r}_S(u)$ for every $u$, in time
$\Ocal(|\Pcal_S|^2n^2)=\Ocal_{\C,r,\varepsilon}(n^2)$ by \cite[Lemma~8]{separability}. So no
$B^r_{G'}(u)$ with $u$ $\varepsilon$-light is $\varepsilon$-heavy, and in both cases the running
time is dominated by the computation of the net.
\end{proof}

\section{Relational structures}\label{app:relational}

We explain how \Cref{thm:main} extends to classes of structures in a finite relational signature.
The proof is that of \Cref{sec:packing}, with the relational Flipper theorem of Przybyszewski and
Toru\'nczyk \cite{flipfork} in place of \Cref{thm:finfliprank} and three facts about relational
flips that are immediate for graphs but not in general; we state them and then indicate the
changes.

\subparagraph*{Setting.}
Fix a finite relational signature $\sigma$. The \emph{Gaifman graph} $\Gaif(M)$ of a
$\sigma$-structure $M$ has as vertices the elements of $M$, two of them adjacent if they occur
together in a tuple of some relation of $M$; distances, balls and $r$-independence in $M$ refer to
$\Gaif(M)$. Weightings, heavy sets and light vertices are defined on the domain exactly as for
graphs. Monadic stability and dependence of a class $\C$ of $\sigma$-structures are defined through
its monadic expansions and the order and independence properties, as in \cite{flipfork}, and
coincide with the graph notions of \Cref{sec:prelims} when $\sigma$ is the signature of graphs.

The flips are those of \cite{flipfork}. For $S\subseteq M$, an \emph{$S$-flip} of $M$ is a
structure $N$ on the same domain, in a possibly different finite relational signature, such that
every relation of $N$ is definable in $M$ by a quantifier-free formula with parameters from $S$,
and vice versa; we write $\flip(M/S)$ for the set of $S$-flips of $M$, and define the flip-metric
$\dist_S$, its balls $B^r_S$, and $\ind^r_S$, $\ind^r_N$ as in \Cref{sec:prelims}, with a supremum
in place of the maximum. When $M$ is a graph, every $S$-definable flip in the sense of
\Cref{sec:prelims}, expanded by unary predicates for its $S$-classes, is an $S$-flip in this sense,
and the two flip-metrics coincide on the domain: the relational metric is at least the graph one,
and conversely any relational $S$-flip of a graph is dominated, in the sense of Gaifman graphs, by
a binary one, whose adjacency on each pair of $S$-classes must distinguish edges from non-edges
unless the adjacency there is constant, and is therefore dominated by an $S$-definable flip. So
the theorem below contains \Cref{thm:main} as the special case of graphs.

\subparagraph*{Three facts about relational flips.}
The first is that flips subsume deletions, exactly as for graphs.

\begin{fact}[isolating flips; {\cite{companion}}]\label{fact:rel-isolating}

Let $M$ be a $\sigma$-structure, $S\subseteq M$ finite and $N\in\flip(M/S)$. There is
$N_0\in\flip(M/S)$ whose Gaifman graph is $\Gaif(N)$ with every edge incident to $S$ removed. In
particular every $s\in S$ is isolated in some $S$-flip, so $B^r_S(s)=\{s\}$.
\end{fact}

The second replaces the trivial bound $|\flip(G/S)|\le\Xi(|S|)$ for graphs. Over a relational
signature there are infinitely many $S$-flips of a structure, of unbounded arity and with pairwise
distinct Gaifman graphs, so the supremum defining $\dist_S$ ranges over an infinite family; the
point is that boundedly many of them, of arity bounded by that of $\sigma$, suffice.

\begin{fact}[finiteness of the flip-metric; {\cite{companion}}]\label{fact:rel-finite}

There is a function $\Xi_\sigma:\N\to\N$, depending only on $\sigma$, such that for every
$\sigma$-structure $M$ and every finite $S\subseteq M$ there is a subfamily
$\flip^*(M/S)\subseteq\flip(M/S)$ of size at most $\Xi_\sigma(|S|)$ such that every
$N\in\flip(M/S)$ satisfies $\Gaif(N')\subseteq\Gaif(N)$ for some $N'\in\flip^*(M/S)$.
Consequently, for all $u,v\in M$ and $r\in\N$,
\[ \dist_S(u,v)=\max_{N'\in\flip^*(M/S)}\dist_{N'}(u,v)
   \qquad\text{and}\qquad
   B^r_S(u)=\bigcap_{N'\in\flip^*(M/S)}B^r_{N'}(u). \]
\end{fact}

The third is the relational Flipper theorem. The flip-separation rank $\rk_r^M(U/S)$ of
\Cref{def:rank} is defined verbatim for structures, with $B^r_{S\cup\{s\}}(u)$ the ball of the
relational flip-metric, and a class is winning for Flipper if for every $r$ the ranks
$\rk_r(M)=\rk^M_r(M/\emptyset)$ are bounded over $M\in\C$.

\begin{fact}[{\cite{flipfork}}]\label{fact:rel-flipper}
A class of $\sigma$-structures is winning for Flipper if and only if it is monadically stable, and
it is flip-flat, in the sense of \Cref{def:flipflat} read for relational flips, if and only if it
is monadically stable.
\end{fact}

\subparagraph*{The theorem.}
Flip-nets and flip-packability are defined for classes of $\sigma$-structures by reading
\Cref{def:flipnet,def:packing} with relational flips.

\begin{theorem}\label{thm:rel-main}
A class of structures in a finite relational signature is flip-packable if and only if it is
monadically stable.
\end{theorem}

\begin{proof}[Proof sketch]
The proof of \Cref{thm:main} goes through mutatis mutandis. Stage~1 of the proof of
\Cref{lem:winning-packing} uses only the rank recursion, the finite additivity of the weighting,
and the triangle inequality of the flip-metric, together with the fact that the parameters played
are isolated in some flip, which is \Cref{fact:rel-isolating}; it applies verbatim. Stage~2 colours
the pairs of an $r$-independent set of the flip-metric by a flip separating them, and needs the
number of colours to be bounded in terms of the size of the parameter set: by
\Cref{fact:rel-finite} the separating flip may be taken in $\flip^*(M/F)$, so $\Xi_\sigma(k)$
colours suffice in place of $\Xi(k)$, and Ramsey's theorem applies as before. The hypothesis that
$\C$ is winning for Flipper is supplied by \Cref{fact:rel-flipper}. For the converse,
\Cref{lem:packing-flat} applies verbatim and yields flip-flatness, which characterises monadic
stability by \Cref{fact:rel-flipper}.
\end{proof}

The counting form of \Cref{cor:counting} extends in the same way, and so does the light-balls
lemma, \Cref{lem:light-balls}, whose proof uses only the triangle inequality and
\Cref{fact:rel-isolating}: for a $2$-flip-packable class of structures, every weighting admits a
bounded parameter set $S$ over which every ball of the flip-metric around a light vertex is light.
This is flip-separability in the flip-metric. What we cannot do for structures is pass from there
to a single flip: the metric conversion of \Cref{fact:conversion}(2) has no known relational
analogue, as its proof rests on a diameter dichotomy for bipartite graphs
\cite[Lemma~10]{separability}, whereas the cells of a relational flip consist of tuples whose pairs
of entries need not form a complete bipartite graph.
Accordingly \Cref{thm:2packing-equiv} is stated for graphs, and the following is open.

\begin{question}[metric conversion]\label{q:conversion}
Is there $c\in\N$ such that for every finite relational signature $\sigma$, every $\sigma$-structure
$M$ and every finite $S\subseteq M$ there are $T\subseteq M$, of size bounded in terms of $|S|$ and
$\sigma$, and $M'\in\flip(M/T)$ with $\dist_S\le c\cdot\dist_{M'}$?
\end{question}

\Cref{fact:rel-finite} shows that $\dist_S$ is computed by boundedly many flips; the question asks
whether one flip suffices up to a constant. A positive answer would extend \Cref{thm:2packing-equiv}
to structures, given the relational forms of \Cref{thm:flipbreak,thm:flipsep}, which are open as
well. Of the remaining entries of \Cref{tab:variants}, the deletion rows are statements about
Gaifman graphs and extend as such, while cliquewidth and shrubdepth are graph parameters; for
shrubdepth, boundedness of $\rk_\infty$ is a candidate relational definition, under which the
corresponding entry extends by the proof of \Cref{prop:app-allm}. Finally, the algorithmic results
of \Cref{sec:algorithmic} rest on the effective form of the Flipper theorem for graphs
\cite{flippergame-full}, and we know of no effective form of \Cref{fact:rel-flipper}; given one,
\Cref{thm:algo} would apply verbatim, with $\Xi_\sigma(k)$ in place of $\Xi(k)$.

\end{document}